\documentclass[10pt, notitlepage]{revtex4-1}
\usepackage[ruled, vlined, linesnumbered]{algorithm2e}
\usepackage{circuitikz}
\usepackage{tikz}
\usetikzlibrary{calc, intersections, through, backgrounds, arrows, positioning, automata, shapes, shapes.geometric, fit}
\usepackage{amsmath}
\usepackage{amssymb}
\usepackage{geometry} 
\usepackage{array}
\usepackage{booktabs}
\usepackage{graphicx}
\usepackage{caption}
\usepackage{listings}
\usepackage{fancyhdr}
\usepackage{yfonts}
\usepackage{mathrsfs}
\usepackage[mathscr]{euscript}
\usepackage{relsize}
\usepackage{bigints}
\usepackage{verbatim} 
\usepackage[english]{babel}
\usepackage{slantsc}
\usepackage{amsthm}
\usepackage{bbm}
\usepackage{placeins}
\usepackage{cancel}
\usepackage{setspace}
\usepackage{mathtools}
\usepackage{kbordermatrix}
\usepackage{blkarray}
\usepackage{multirow}
\usepackage{bigdelim}
\usepackage{ulem}
\usepackage{xcolor}

\DeclareMathAlphabet{\mathpzc}{OT1}{pzc}{m}{it}

\usepackage{calligra}
\usepackage{rotunda}
\usepackage{mflogo}
\usepackage{frcursive}
\usepackage{suetterl}

\input Sanremo.fd

\input Rothdn.fd

\newcommand{\mfr}[1]{\text{\fontfamily{frc}\selectfont\bfseries\footnotesize#1}\,}

\newcommand{\msuet}[1]{\text{\Large\suetterlin#1}}

\newfont{\rsfsten}{rsfs10 scaled 1200}

\newcommand{\n}{\vspace{\baselineskip}}
\newcommand{\Z}{\mathbb{Z}}

\newcommand{\sm}{\smallskip} 
\newcommand{\m}{\medskip} 

\newcommand{\bra}[1]{\left\langle\,#1\,\right\rvert} 
\newcommand{\ket}[1]{\left\lvert\,#1\,\right\rangle} 

\newcommand{\dg}{\dagger}

\newcommand{\mim}{\Leftrightarrow} 

\newtheorem{thm}{Theorem}[section]
\newtheorem{cor}{Corollary}[thm]
\newtheorem{lem}[thm]{Lemma}

\newtheoremstyle{propositional}
{10pt}
{10pt}
{
	\addtolength{\linewidth}{-2.0em}
	\parshape 1 2.0em \linewidth} 
{}
{\sc}
{:}
{.5em}
{}
\makeatother

\theoremstyle{propositional}

\newtheoremstyle{definitive}
{10pt}
{10pt}
{
	\addtolength{\linewidth}{-2.0em}
	\parshape 1 0.0em \linewidth} 
{}
{\sc}
{:}
{.5em}
{}
\makeatother

\theoremstyle{definitive}
\newtheorem*{defi}{Definition}

\newcounter{numb}
\newcounter{bean}
\newcounter{count}
\newcounter{count2}
\newcounter{count3}
\newenvironment{problems}{\begin{list}{\arabic{numb}.}{\usecounter{numb}
			\setlength{\leftmargin}{20pt}
			\setlength{\labelwidth}{15pt}
			\setlength{\labelsep}{5pt}
			\setlength{\itemsep}{ 15.0pt plus 2.5pt minus 10.0pt}
	}}{\end{list}}

\makeatletter

\newlength\tdima
\newcommand\tabfill[1]{%
	\setlength\tdima{\linewidth}%
	\addtolength\tdima{\@totalleftmargin}%
	\addtolength\tdima{-\dimen\@curtab}%
	\parbox[t]{\tdima}{#1\ifhmode\strut\fi}}

\newcommand\mytabs{\hspace*{5cm}\=\hspace{1cm}\=\hspace{2cm}}

\makeatother

\definecolor{gainsboro}{RGB}{220,220,220}
\definecolor{textclr}{RGB}{131,148,150}
\definecolor{monokai}{HTML}{272822}
\definecolor{darkb}{RGB}{0, 43, 54}

\tikzset{elliptic state/.style={draw,ellipse}}

\DontPrintSemicolon

\newcommand*\justify{%
  \fontdimen2\font=0.4em
  \fontdimen3\font=0.2em
  \fontdimen4\font=0.1em
  \fontdimen7\font=0.1em
  \hyphenchar\font=`\-
}

\begin{document}
	
	\raggedbottom
	
	\title{Constructing Driver Hamiltonians for Optimization Problems with Linear Constraints}
	\author{Hannes Leipold$^{1,2}$, Federico M. Spedalieri$^{1,3}$}
    \affiliation{
		$^1$Information Sciences Institute, University of Southern California, Marina del Rey, CA 90292, USA \\
		$^2$Department of Computer Science, University of Southern California, Los Angeles, CA 90089, USA \\
		$^3$Department of Electrical and Computer Engineering, University of Southern California, Los Angeles, CA 90089, USA
	}

	\date{\today}
	
	\begin{abstract}
		Recent advances in the field of adiabatic quantum computing and the closely related field of quantum annealing have centered around using more advanced and novel Hamiltonian representations to solve optimization problems. One of these advances has centered around the development of driver Hamiltonians that commute with the constraints of an optimization problem - allowing for another avenue to satisfying those constraints instead of imposing penalty terms for each of them. In particular, the approach is able to use sparser connectivity to embed several practical problems on quantum devices in comparison to the standard approach of using penalty terms. However, designing the driver Hamiltonians that successfully commute with several constraints has largely been based on strong intuition for specific problems and with no simple general algorithm for generating them for arbitrary constraints. In this work, we develop a simple and intuitive algebraic framework for reasoning about the commutation of Hamiltonians with linear constraints - one that allows us to classify the complexity of finding a driver Hamiltonian for an arbitrary set of linear constraints as NP-Complete. Because unitary operators are exponentials of Hermitian operators, these results can also be applied to the construction of mixers in the Quantum Alternating Operator Ansatz (QAOA) framework. 
	\end{abstract}
	
	\maketitle
	
	\setcounter{section}{0}

	\section{Introduction}
	
	\quad Quantum annealing has been proposed as a heuristic method to exploit quantum mechanical effects in order to solve discrete optimization problems. Typically, these problems require optimizing a quadratic cost function subject to a set of linear constraints. The usual approach to treating these constraints consists of adding them to the cost function as penalty terms, thereby transforming the constrained optimization into an unconstrained one. This approach has some drawbacks, including an increase in the required resources (i.e., higher connectivity and increased dynamical range for the parameters that define the instance). In Ref.\cite{hen2016quantum}, the authors introduced the idea of constrained quantum annealing (CQA), that uses specially tailored driver Hamiltonians for a specific set of constraints. These tailored Hamiltonians have several advantages, such as reducing the size of the search space of the problem and reducing the number of needed interactions to implement the annealing protocol. At the heart of the approach is the idea that a Hamiltonian which commutes with the operator embedding of the constraints and starts within the feasible space of configurations will remain in it throughout the evolution.

	\m 
	
	\quad While Ref.\cite{hen2016quantum} looked primarily at commuting with a single global constraint, the work in Ref.\cite{hen2016driver} focused on finding driver Hamiltonians for several constraints. Under special conditions, the authors were able to construct appropriate driver Hamiltonians for several optimization problems of practical interest. In this paper, we ask and answer the general question of, given a set of arbitrary linear constraints, can we  construct a driver Hamiltonian which will commute with the operator embedding of the constraints? Our main result is that this problem is NP-Complete -- answering a question posited originally in Ref.\cite{hen2016driver}. Along the way, we will derive a simple formula for describing the commutation relation and exploit it to understand many facets of CQA. These results can be naturally applied to the Quantum Alternating Operator Ansatz (QAOA) framework~\cite{hadfield2017quantum}, where one of the central tasks is to find  mixer operators that connect feasible solutions of a constrained optimization problem. Since unitary operators are exponentials of Hermitian operators (for every unitary matrix $ U $, there exists a Hermitian matrix $ H $ such that $ U = e^{i H } $) and therefore the existence of a unitary matrix with this commutative property necessitates the existence of the Hermitian matrix with the same commutative property (since $ [ e^{ i H } , \hat{C} ] = \sum_{k=0}^{\infty} [ \left( i H \right)^{ k } , \hat{C} ] / k! $), our results directly translate into this setting as well.
	
	\m 
	
	\quad  The paper is organized as follows. In Section~\ref{sec:background} we review the basic ideas behind constrained quantum annealing. In Section~\ref{algcondition}, we derive a simple algebraic condition for the commutation relation of the driver Hamiltonian and the constraint operators. For many practical applications, Hamiltonians with bounded weight interaction (local) terms are often desired; in Section~\ref{sec:bounded} we show how brute forcing the simple algebraic condition from Section~\ref{algcondition} to find driver Hamiltonians of bounded weight. In Section~\ref{sec:probdef}, we introduce several variations of the problem \texttt{ILP-QCOMMUTE}, the problem of finding a Hermitian matrix that will commute with the constraint operators. We also reduce the \texttt{EQUAL SUBSET SUM} problem to the problem \texttt{ILP-QCOMMUTE} (and in Appendix~\ref{sec:redoxold} we show it for the special case of binary valued linear constraints). The \texttt{EQUAL SUBSET SUM} problem is known to be NP-Complete, thus proving our main result. We define a related complexity class, \texttt{ILP-QCOMMUTE-k-LOCAL}, about finding a Hermitian matrix that commutes with the constraints and has interaction terms up to $ k $ weight. \texttt{ILP-QCOMMUTE-k-LOCAL} is in \texttt{P} by the approach detailed in Section~\ref{sec:bounded}. Other questions of interest, such as finding a driver Hamiltonian that has full reachability over a constrained space for a CQA protocol, will be addressed through our formulation in Section~\ref{sec:reachability}. We conclude with a discussion of the significance of our result and open problems related to what we have shown in this work.
	
	\section{Background}\label{sec:background}
	
	\quad In the quantum annealing (QA) framework, gradually decreasing quantum fluctuations are used to traverse the barriers of an energy landscape in the search of global minima to complicated cost functions \cite{kadowaki1998quantum,farhi2000quantum}. For an overview of these approaches, we refer the reader to Ref.\cite{albash2018adiabatic}. Quantum annealing has gained traction for combinatorial optimization \cite{farhi2001quantum,santoro2006optimization,bian2016mapping,hauke2020perspectives} as a way to solve hard optimization problems faster and, more recently, for machine learning \cite{adachi2015application,mott2017solving,amin2018quantum,biamonte2017quantum,li2018quantum,kumar2018quantum} as a way to naturally sample desired probability distributions quickly. In the case of solving an optimization problem, the problem is encoded in the Hamiltonian $ H_{p} $ such that the ground state is the optimum solution. Usually this is readily done by expressing the problem as an Ising model, a model for spin glasses\cite{brush1967history,sherrington1975solvable,castellani2005spin,troyer2005computational}. 
	
	\m 
	
	\quad Once the problem Hamiltonian is described, the QA framework prescribes an evolution to the final Hamiltonian from some readily preparable Hamiltonian $ H_{d} $ - usually through a linear interpolation of $ H_{d} $ and $ H_{p} $:
	\begin{align}
	H(s) = s \, H_{p} + (1-s) \, H_{d}, 
	\end{align}
	
	where there is a continuous smooth function $ s(t) $ for $ t \in [ 0, T ] $ such that $ s(0) = 0 $ and $ s(T) = 1 $. If the process is varied slowly enough, the adiabatic theorem ensures that the wave function of the system will be close to the instantaneous ground state of the system for any $ s $ and therefore any $ t $. By the adiabatic theorem, if the total time $ T $ that the system is evolved for is large compared to the inverse of the minimum gap squared, then the wavefunction of the system will be close to the ground state of $ H_{p} $. For the purpose of our presentation here, we restrict our focus to the case of binary linear optimization problems, a heavily studied optimization class. Specifically, we can consider a set of linear constraints $ \mathcal{C} = \{ C_{1}, \ldots, C_{m} \} $. Suppose that the solutions to the optimization problem are subject to constraints $ C_{i} $ such that a solution state $ x \in \{ 0, 1 \}^{n} $ satisfies $ C_{i} (x) = \sum_{j} c_{ij} x_{i} = b_{i} $ for some $ b_{i} $. Because $ C_{i} $ is a simple linear function, we can associate a vector $ \vec{c}_{i} $ with it such that $ C_{i}( x ) = \vec{c}_{i} \cdot \vec{x} $ where $ \vec{c}_{i} \in \Z^{n} $. When referring to constraints throughout this paper, we are referring specifically to linear constraints, for which our main results are pertinent to.
	
	\subsection{Constraint Quantum Annealing for Integer Linear Programming}
	
	\quad We use the ordinary embedding of \textit{binary} variables $ x_{i} \in \{ 0, 1 \} $ in the computational basis for quantum annealing, such that $ \vec{x} \in \{ 0, 1 \}^{n} $ is represented by $ \ket{ \vec{x} } = \ket{ x_{1} } \ldots \ket{ x_{n} }  \in \mathbb{C}^{2^{n}} $ (i.e. $ \sigma_{i}^{z} \ket{ \vec{x} } = \left( 1 - 2 \, x_{i} \right) \ket{ \vec{x} } $) and the final Hamiltonian is diagonal in the computational basis so that we can read off a solution by measuring in that basis. Following the framework of CQA, given a constraint $ C(x) = \vec{c} \cdot \vec{x} $, we associate $ C $ with an embedded constraint operator $ \hat{C} = \sum_{i=1}^{n} c_{i} \sigma_{i}^{z} $. Let us consider the case of a single constraint - $ C = (1,\ldots,1) $ - over $ n $ variables. This is also the first case presented in Ref.\cite{hen2016quantum}. It is simple to check that $ H_{d} = \sum_{i=1}^{n-1} (\sigma_{i}^{x} \sigma_{i+1}^{x} + \sigma_{i}^{y} \sigma_{i+1}^{y}) $ commutes with the constraint embedded operator $ \hat{C} = \sum_{i=1}^{n} \sigma_{i}^{z} $. For example, this type of constraint may arise in graph partitioning, since the partitions must split the graph into equal size. For the \texttt{Graph Partition} problem, one is given a graph $ G $ and is asked to partition the vertices $ V $ into equal subsets such that the number of edges between the two is minimized. In terms of the Ising model, we can consider a collection of $ n $ qubits, such that $ \ket{ 0 } $ ($\ket{1}$) for qubit $ i $ represents placing vertex $ v_{i} $ in partition 1 (2). As such, we design a penalty Hamiltonian $ H_{p} $ and a driver Hamiltonian $ H_{d} $ such that the final state will be a solution to the graph partitioning problem. Assuming the transverse field driver Hamiltonian - $ H_{d} = \sum_{i=1}^{n} \sigma_{i}^{x} $ - a simple penalty Hamiltonian can be:
	\begin{align}
	H_{p} = \sum_{(i,j) \in E} \left( \mathbbm{1} - \sigma_{i}^{z} \sigma_{j}^{z} \right) + \alpha \, \left( \sum_{i=1}^{n} \sigma_{i}^{z} \right)^{2},
	\end{align}
	
	where the first term assigns a positive potentiality to each edge that connects vertices across the partitions and the second term is the constraint operator squared. 
	
	\m 
	
	\quad In general the penalty factor $ \alpha $ must be greater than $ \text{min}(2 \, d_{m}, n)/8$ where $ d_{m} $ is the maximal degree of $ G $ \cite{lucas2014ising}. Note that the term $ \left( \sum_{i=1}^{n} \sigma_{i}^{z} \right)^{2} $ is $ \hat{C}^{2} $ and requires $ n^{2} $ two body interaction terms to implement. However, if we choose our $ H_{d} $ such that $ [ H_{d}, \hat{C} ] = 0 $, then we can use the simpler penalty Hamiltonian:
	\begin{align}
	H_{p} = \sum_{(i,j) \in E} \left( \mathbbm{1} - \sigma_{i}^{z} \sigma_{j}^{z} \right). 
	\end{align}
	
	Note that since $ H_p $ is diagonal in the spin-z basis, it trivially commutes with the constraints.
	
	\m
	\quad One benefit to this construction is that the driver $ H_{d} = \sum_{i=1}^{n-1} (\sigma_{i}^{x} \sigma_{i+1}^{x} + \sigma_{i}^{y} \sigma_{i+1}^{y}) $, for example, commutes with $ \hat{C} $ and only requires $ n-1 $ two body interaction terms to implement. As such, the total number of two body terms required to solve the problem can be greatly reduced by using driver Hamiltonians beyond the transverse field if they commute with a set of constraints. As long as the initial wavefunction has an expected value of $ n / 2 $ for $ \hat{C} $ ($ n / 2 + 1 $ or $ n / 2 - 1 $ if $ n $ is odd), the wavefunction will remain in the subspace with this expected value for the entirety of the anneal. As an example of \texttt{Graph Partition} that we will return to later, consider a graph with 4 vertices $ V = \{ v_{1}, v_{2}, v_{3}, v_{4} \} $, connected into a single path by edges $ E = \{ e_{1}, e_{2}, e_{3} \} $ with $ e_{1} = (v_{1}, v_{2}), e_{2} = (v_{2}, v_{3}), e_{3} = (v_{3}, v_{4}) $. Then in this case, $ H_{d} = (\sigma_1^x \sigma_2^x + \sigma_1^y \sigma_2^y) + (\sigma_2^x \sigma_3^x + \sigma_2^y \sigma_3^y) + (\sigma_3^x \sigma_4^x + \sigma_3^y \sigma_4^y) $. Since we are interested in an even partition, the starting state should have an equal number of 1s and 0s. For example, $ \ket{ 0 0 1 1 } $ is in the correct subspace. It is important to note that driver Hamiltonians are constructed irrespective of the value that the constraint is set to, since the eigenvalue of the constraint operator with respect to the initial wavefunction will determine what value is preserved during the anneal. For a wavefunction that is not an eigenvector of a constraint operator, the commuting property of the driver Hamiltonian with the constraint operator will mean that the projection of the wavefunction onto each specific eigenspace will evolve independently of the rest of the wavefunction.
	
	\m 
	\quad For the purpose of adiabatic quantum computing, the initial wavefunction of the system has to be in the ground state of the initial Hamiltonian, while a general driver Hamiltonian from this construction can be highly nontrivial if there are many constraints. Therefore, specifying an alternative initial Hamiltonian for CQA is also a major area of research, since we want the initial Hamiltonian to have a ground state that is easy to prepare. One approach to overcome this is to use an initial Hamiltonian diagonal in the computational basis and linear on the $\sigma_z$ operators, that has as its ground state a \textit{specific} solution in the feasible space in the spin-z basis, and then evolve from this initial Hamiltonian to the driver Hamiltonian (whose ground state will have support on all or a subset of the feasible space). There are many hard problems for which finding a feasible solution is simple. For example, it is straightforward to find a single partition for \texttt{Graph Partition}, but it is hard to find the \textit{best} partition. As such, a useful avenue for exploiting CQA is in the case where linear constraints specify a nontrivial feasibility space, but one where finding a \textit{nonoptimum} element in the feasibility space is still tractable. 
	
	\m 
	\quad The work in Ref.\cite{hen2016driver} extended the framework for cases where a driver should commute with multiple constraint operators. In particular, given a set $ \mathcal{C} $, they consider finding a Hamiltonian $ H_{d} $ such that:
	\begin{align}
	[H_{d}, \hat{C}_{j} ] = 0, \; C_{j} \in \mathcal{C}
	\end{align}
	
	\quad As they note, in general, tailoring driver Hamiltonians for a set $ \mathcal{C} $ can be difficult. In this paper, we answer specifically the computational complexity of such a task by reducing an NP-Complete problem to \texttt{ILP-QCOMMUTE}. We also discuss the related task of finding $ H_{d} $ such that it can reach \textit{every} state in the solution space, but reaches no state outside the solution space in Section~\ref{sec:reachability}. This result in some ways may appear intuitive, since describing the feasible space of $ \mathcal{C} $ is hard and knowing a Hamiltonian that would keep a wavefunction within this space - and only this space - should require some characterization of it. Simply knowing that a nontrivial $ H_{d} $ exists at all for a NP-Complete feasibility problem allows one to recognize that the problem should have at least two solutions for some set of values that each constraint is set to, even if one does not have a token to prove it.
	
	\section{An Algebraic Condition for Commuting with Linear Constraints}\label{algcondition}
	
	\quad Consider the problem to find Hamiltonian drivers that have, as their eigenvectors, support over the possible values that satisfy the given linear constraints. In the most general sense, we consider constraints of the form:
	\begin{align}
	\hat{C} = \sum_{i=1}^{n} c_{i} \, \sigma_{i}^{z}, \; \vec{c} \in \Z^{n},
	\end{align}
	
	with a constraint value $ b $ that corresponds to one of the energy levels of $ \hat{C} $. Often problems of practical interest can be captured in the restricted case where $ \vec{c} \in \{0,1\}^{n} $ or $ \vec{c} \in \{-1,0,1\}^{n} $. 
	
	\m
	
	\quad Consider the linear transformation $ [ M, \sigma^{z} ] $ that maps any two by two matrix $ M $ to a new two by two matrix $ M' $, by commuting the matrix with $\sigma^{z}$. This transformation has two obvious eigenmatrices - $ \mathbbm{1} $ and $ \sigma^{z} $ - that span the kernel of the transformation. One can easily verify that $ \sigma^{+} $ ($ \ket{ 0 }\bra{ 1 } $) and $ \sigma^{-} $ ($ \ket{ 1 }\bra{ 0 } $) are also \textit{eigenmatrices} of this transformation, with eigenvalue $ 2 $ and $ -2 $ respectively. Together these four eigenmatrices and their eigenvalues describe the spectrum decomposition of the transformation. We exploit this fact to find a simple algebraic formula for expressing the commutation of a general Hamiltonian with a linear constraint. It is easy to verify that for $ H $ over $ n $ qubits, if $ \text{Tr}_{ 1, \ldots i-1 , i+1, \ldots n }[ H ] = \sigma^{\pm} $, then $ [ H, \sigma_{i}^{z} ] = \pm 2 \, H $. 
	
	\m 
	
	\quad Given any complete basis for a single qubit system, we can extend that basis to define a basis over $ n $ qubits. Doing this with the found eigenmatrices defines a basis $ \{ \mathbbm{1}, \sigma^{z}, \sigma^{+}, \sigma^{-} \}^{\otimes n} $. Note as well that $ \left( \alpha_{j} \sigma_{i}^{\pm} \right)^{\dg} = \alpha_{j}^{\dg} \sigma_{i}^{\mp} $ for $ \alpha_{j} \in \mathbb{C} $. This suggests a simple representation in which a Hermitian matrix is defined by its nonzero terms over this basis. Then \textit{any} Hermitian matrix can be written in the form:
	\begin{align}
	H &= \sum_{ ( \vec{y_{j}}, \vec{v_j}, \vec{w_j} ) \in \Delta( \mathscr{Y},  \; \mathscr{V}, \; \mathscr{W} ) } \alpha_{j} \bigotimes_{i = 1}^{n} \left( \sigma^{z} \right)^{y_{ji}} \left( \sigma^{+} \right)^{v_{ji}} \left( \sigma^{-} \right)^{w_{ji}} \nonumber \\ 
	&\phantom{= } + \sum_{ ( \vec{y_{j}}, \vec{v_j}, \vec{w_j} ) \in \Delta( \mathscr{Y}, \; \mathscr{V}, \; \mathscr{W} ) } \alpha_{j}^{\dg} \bigotimes_{i = 1}^{n} \left( \sigma^{z} \right)^{y_{ji}} \left( \sigma^{+} \right)^{w_{ji}} \left( \sigma^{-} \right)^{v_{ji}}, \label{hexpand}
	\end{align}
	
	where $ \mathscr{Y} = \{ \vec{y_{1}}, \ldots, \vec{y_{ r }} \} $ with $ \vec{y}_{i} \in \{0,1\}^{n} $, $ \mathscr{V} = \{ \vec{v_{1}}, \ldots, \vec{v_{ r }} \} $ with $ \vec{v}_{i} \in \{0,1\}^{n} $, and $ \mathscr{W} = \{ \vec{w_{1}}, \ldots, \vec{w_{ r }} \} $ with $ \vec{w}_{i} \in \{0,1\}^{n} $ are such that the corresponding $ \alpha_{i} \neq 0 $. Here $ \Delta( \mathscr{Y}, \mathscr{V}, \mathscr{W} ) = \{ ( \vec{y}_{i}, \vec{v}_{i}, \vec{w}_{i} ) | y_{i} \in \mathscr{Y}, v_{i} \in \mathscr{V}, w_{i} \in \mathscr{W} \} $, where $ \Delta $ simply takes any indexed element sets and creates the set of the index-wise confederated tuples. A tuple $ (\vec{y}_{i}, \vec{v}_{i}, \vec{w}_{i} ) $ specifies the indices in which we chose $ \sigma^{z} $, $ \sigma^{+} $, or $\sigma^{-}$ for each nonzero term. Once that choice is made, hermiticity demands the corresponding second term seen in Eq.~\ref{hexpand} to be part of the Hamiltonian as well. However, there are restrictions on what vectors can be chosen. Specifically, $ \vec{y}_{i} \cdot \vec{w}_{i} = 0 $ and $ \vec{y}_{i} \cdot \vec{v}_{i} = 0 $, since choosing $ \sigma^{z} $ and a $ \sigma^{\pm} $ would actually mean selecting $ \sigma^{\pm} $ with a new coefficient $ - \alpha_{j} $ instead. Likewise, it should also be the case that $ \vec{v}_{i} \cdot \vec{w}_{i} = 0 $ -- otherwise the term would be equivalent to two terms with the same coefficient halved and one taking the term $ \sigma^{z} $, the other taking the identity term. As such, this added constraints on $ \Delta(\mathscr{Y}, \mathscr{V}, \mathscr{W}) $ so that the representations are \textit{unique}, which must be enforced before applying the theorem below because the uniqueness of the basis will be actively used. As an example, consider the driver Hamiltonian discussed in the previous section: $ H_{d} = \sum_{i = 1}^{n - 1} ( \sigma_{i}^{x} \sigma_{i+1}^{x} + \sigma_{i}^{y} \sigma_{i+1}^{y} ) = 2 \, \sum_{i = 1}^{n - 1} ( \sigma_{i}^{+} \sigma_{i+1}^{-} + \sigma_{i}^{-} \sigma_{i+1}^{+} ) $. For this Hamiltonian, $ \mathscr{Y} = \{ \vec{0}, \ldots, \vec{0} \}, \mathscr{V} = \{ \vec{e}_{1}, \ldots, \vec{e}_{n-1}  \}, \mathscr{W} = \{ \vec{e}_{2}, \ldots, \vec{e}_{n} \} $ -- where $ \vec{e}_{i} $ refers to the standard basis vectors. While the notation is somewhat awkward, it becomes useful for expressing our first major result:
	
	\begin{thm}[Algebraic Condition for Commutativity] \label{algconcom}
		A Hermitian Matrix $ H $ commutes with a linear constraint $ C $ if and only if $ \vec{c} \cdot (\vec{v}_{j} - \vec{w}_{j}) = 0 $ for all $ \vec{v}_{j}, \vec{w}_{j} \in \Delta( \mathscr{V}, \mathscr{W} ) $. 
	\end{thm}
	
	\begin{proof}
		Using the form for $ H $ we introduced earlier, we can see that:
		\begin{align}
		\left[ H, \hat{C} \right] &= \sum_{ ( \vec{y_{j}}, \vec{v_j}, \vec{w_j} ) \in \Delta( \mathscr{Y}, \; \mathscr{V}, \; \mathscr{W} ) } \left[ \alpha_{j} \bigotimes_{i = 1}^{n} \left( \sigma^{z} \right)^{y_{ji}} \left( \sigma^{+} \right)^{v_{ji}} \left( \sigma^{-} \right)^{w_{ji}}, \sum_{k=1}^{n} c_{k} \, \sigma_{k}^{z} \right] \nonumber \\ 
		&\phantom{= } + \sum_{ ( \vec{y_{j}}, \vec{v_j}, \vec{w_j} ) \in \Delta( \mathscr{Y}, \; \mathscr{V}, \; \mathscr{W} ) }  \left[ \alpha_{j}^{\dg} \bigotimes_{i = 1}^{n} \left( \sigma^{z} \right)^{y_{ji}} \left( \sigma^{+} \right)^{w_{ji}} \left( \sigma^{-} \right)^{v_{ji}}, \sum_{k=1}^{n} c_{k} \, \sigma_{k}^{z} \right] \nonumber \\
		&= \sum_{ ( \vec{y_{j}}, \vec{v_j}, \vec{w_j} ) \in \Delta( \mathscr{Y}, \; \mathscr{V}, \; \mathscr{W} ) } 2 \, \alpha_{j} \left( \sum_{k=1}^{n} c_{k} (v_{jk} - w_{jk}) \right) \bigotimes_{i = 1}^{n} \left( \sigma^{z} \right)^{y_{ji}} \left( \sigma^{+} \right)^{v_{ji}} \left( \sigma^{-} \right)^{w_{ji}}  \nonumber \\ 
		&\phantom{= } + \sum_{ ( \vec{y_{j}}, \vec{v_j}, \vec{w_j} ) \in \Delta( \mathscr{Y}, \; \mathscr{V}, \; \mathscr{W} ) }  2 \, \alpha_{j}^{\dg} \left( \sum_{k=1}^{n} c_{k} (w_{jk} - v_{jk}) \right) \bigotimes_{i = 1}^{n} \left( \sigma^{z} \right)^{y_{ji}} \left( \sigma^{+} \right)^{w_{ji}} \left( \sigma^{-} \right)^{v_{ji}} \nonumber \\
		&= \sum_{ ( \vec{y_{j}}, \vec{v_j}, \vec{w_j} ) \in \Delta( \mathscr{Y}, \; \mathscr{V}, \; \mathscr{W} ) } 2 \, \alpha_{j} \, \vec{c} \cdot \left( \vec{v}_{j} - \vec{w}_{j} \right) \bigotimes_{i = 1}^{n} \left( \sigma^{z} \right)^{y_{ji}} \left( \sigma^{+} \right)^{v_{ji}} \left( \sigma^{-} \right)^{w_{ji}}  \nonumber \\ 
		&\phantom{= } - \sum_{ ( \vec{y_{j}}, \vec{v_j}, \vec{w_j} ) \in \Delta( \mathscr{Y}, \; \mathscr{V}, \; \mathscr{W} ) }  2 \, \alpha_{j}^{\dg} \, \vec{c} \cdot \left( \vec{v}_{j} - \vec{w}_{j} \right) \bigotimes_{i = 1}^{n} \left( \sigma^{z} \right)^{y_{ji}} \left( \sigma^{+} \right)^{w_{ji}} \left( \sigma^{-} \right)^{v_{ji}} \label{eq:comeq}
		\end{align}
		
		Since the set defined by the tuples $ \{ ( \vec{y}_{j}, \vec{v}_{j}, \vec{w}_{j} ) \} $ is a linearly independent set Eq.~\ref{eq:comeq} $ = 0 $ iff $ \vec{c} \cdot ( \vec{v_{j}} - \vec{w_{j}} ) = 0 $ for all $ j $.
		
	\end{proof}
	
	Consider the example discussed in Section~\ref{sec:background} in the context of Theorem~\ref{algconcom}. It is easy to see that $ v_{1} = (1,0,0,0), w_{1} = (0,1,0,0), v_{2} = (0,1,0,0), w_{2} = (0,0,1,0), v_{3} = (0,0,1,0), w_{3} = (0,0,0,1) $ would satisfy Eq.~\ref{eq:comeq}, defining $ H_{d} = \sum_{i=1}^{3} \sigma_{i}^{+} \sigma_{i+1}^{-} + \sigma_{i}^{-} \sigma_{i+1}^{+} $. Note that more vector pairs will also satisfy the condition, for example $ v_{4} = (0,0,0,1) $ and $ w_{4} = (1,0,0,0) $.
	
	\section{Bounded weight drivers}\label{sec:bounded}
	
	\quad The algebraic condition of Theorem~\ref{algconcom} can be used as a starting point to understand several features of CQA. One of them is motivated by the fact that  actual implementations of quantum annealing will likely impose a bound on the weight of the driver operators available. Hence, given a set of linear constraints, we could restrict our search for commuting drivers to those with bounded weight. 
	
	\m 
	
	\quad Consider a set $ \mathcal{C} = \{ C_{1}, \ldots, C_{m} \} $ of linear constraints on $ n $ variables, and let $ C^{M} $ be the $ m \times n $ matrix with coefficients $ c_{ij} $ (recall $ C_i (x)  = \sum_j c_{ij} x_j $). Then, Theorem~\ref{algconcom} tells us that there is a non-diagonal driver commuting with all the constraints, if and only if, the linear system $ C^{M} \, \vec{u} = \vec{0} $, where $ \vec{u} = \vec{v} - \vec{w} \in \{ -1, 0, 1 \}^{n} $. Furthermore, we can see that the number of components of $ \vec{u} $ that are non-zero is a lower bound for the weight of such a driver (the weight could be higher if we allow $ \sigma^z $ operators acting on the variables associated with the vanishing components of $ \vec{u} $). This leads to a simpler analysis of the case of bounded weight drivers.
	
	\m 
	
	\quad Assume that we are only allowed weight $ k $ drivers, and we further impose the condition that they are constructed from $ k $ 1-qubit operators that act non trivially in the computational basis (in our case, it means these are chosen from $ \{\sigma^+ , \sigma^- \} $). Then, the number of such operators is $ 2^{k-1} \binom{n}{k} $. For fixed $ k $, this is polynomial in $ n $, and so it is tractable to check all the possible vectors $ \vec{u} $ to find which ones satisfy the condition $ C^{M} \, \vec{u} = 0$. From those that do, we can construct the corresponding weight $ k $ driver that will commute with all the constraints, by assigning $ \sigma^+ $ to qubit $ i $ if $ u_i = 1 $, and 
	$ \sigma^- $ if $ u_i = -1 $. This is a simple but very useful result in practice: brute force searching for solutions to the condition from Theorem~\ref{algconcom} finds all possible driver Hamiltonians that commute with the embedded constraint operators up to a certain weight in time polynomial on the system size.

	\quad The simplest and most relevant case for practical applications is that of 2-local operators, since 2-body interactions are more easily engineered than higher order ones. In this case the condition of Theorem~\ref{algconcom} implies an even simpler characterization of when a set of constraints allow for commuting drivers. 
	\begin{cor}
        Let $ \mathcal{C} = \{ C_{1}, \ldots, C_{m} \}$ be a set of linear constraints and $ C^{M} $ the associated matrix of coefficients. Then a 2-local driver that commutes with all the constraints exists, if and only if, the matrix $ C^{M} $ has a pair of columns that are either equal, or opposite.
    \end{cor}
	
    \quad 2-local means that $ \vec{u} $ in the condition $ C^{M} \, \vec{u} = 0 $ has only 2 non zero components, and we can take these to be $(1,1)$ or $(1,-1)$, since the other two possibilities would produce the corresponding Hermitian conjugates. Since multiplying a matrix by a column vector results in a linear combination of the columns of the matrix, with the coefficients given by the vector components, the condition $ C^{M} \, \vec{u} = 0 $ will state that $ C^M $ has two columns that are the same (if the non zero components of $ \vec{u} $ are $ (1,-1) $), or are opposites (if they are $ (1,1) $). To any distinct pair of columns that satisfy one of these conditions we can associate a distinct weight-2 driver that commutes with the constraints, so the maximum number of such possible weight-2 drivers is $ \binom{n}{2} $.

	\section{The Problem \texttt{ILP-QCOMMUTE}}\label{sec:probdef}
	
	\quad Having found a simple algebraic condition for expressing the commutation relationship of \textit{any} Hermitian matrix with a linear constraint, we wish to exploit that fact here to find what the general complexity of knowing the existence of such a Hermitian matrix is. We consider the following problem:
	
	\begin{defi}[ILP-QCOMMUTE] 
	Given a set $ \mathcal{C} = \{ C_{1}, \ldots, C_{m} \} $ of linear constraints such that $ \hat{C_{i}} = \sum_{j = 1}^{n} c_{ij} \sigma_{j}^{z} $, over a space $ \mathbb{C}^{2^{n}} $ with $ c_{ij} \in \Z $, is there a Hermitian Matrix $ H $, with $ \mathscr{O} \left( \text{poly}( n ) \right) $ nonzero coefficients over a basis $ \{ \chi_1, \chi_2, \chi_3, \chi_4 \}^{\bigotimes n} $, such that $ \left[ H, \hat{C_{i}} \right] = 0 $ for all $ \hat{C_{i}} $ and $ H $ has at least one off-diagonal term in the spin-z basis?
	\end{defi}
	
	\quad Solving this problem would be useful for constructing Hamiltonian drivers for quantum annealing. We can also define \texttt{0-1-LP-QCOMMUTE} as the binary version, where $ c_{ij} \in \{ 0 , 1 \} $, and also \texttt{\{-1,0,1\}-LP-QCOMMUTE}, where $ c_{ij} \in \{-1,0,1\} $, the type of coefficients used when representing problems like 1-in-3 3-SAT as an ILP. One of the central results of this paper is that these problems are NP-Complete\cite{cook1971complexity}\cite{karp1975computational}, which can be shown by reducing them to the \texttt{EQUAL SUBSET SUM} problem\cite{karp1972reducibility}. This reduction is simple and straightforward for \texttt{ILP-QCOMMUTE}, and we discuss it in this section to give a sense of the connection between the two problems. However, this proof is not enough to imply that \texttt{0-1-LP-QCOMMUTE} and \texttt{\{-1,0,1\}-LP-QCOMMUTE} are also NP-Complete, since they could both very well be easier subclasses of \texttt{ILP-QCOMMUTE}. That is not the case, but the proof for \texttt{0-1-LP-QCOMMUTE} is more involved (and rather tedious) and thus is presented in the Appendix~\ref{sec:redoxold}.
	
	
	\begin{defi}[EQUAL SUBSET SUM]
		Given a set $ S = \{ s_{1}, s_{2}, \ldots, s_{n} \} $, with $ s_{i} \in \Z^{+} $, are there two non-empty disjoint subsets, $ A, B $ such that $ \sum_{a_{i} \in A} a_{i} = \sum_{b_{i} \in B} b_{i} $?
	\end{defi}
	
	\quad The \texttt{EQUAL SUBSET SUM} problem is known to be NP-Complete\cite{woeginger1992equal}. We map an instance of the \texttt{EQUAL SUBSET SUM} problem to the \texttt{ILP-QCOMMUTE} problem; the former defined over a set $ S = \{ s_{1}, s_{2}, \ldots, s_{n} \} $, with $ s_{i} \in \Z^{+} $. Consider the constraint operator defined by $ \hat{C} = \sum_{i=1}^{n} s_{i} \sigma_{i}^{z} $, and the vector $\vec{s} = (s_1,\ldots,s_n)$. Suppose we can find vectors $ \vec{v}, \vec{w} $ with binary components, such that $ \vec{s} \cdot \left( \vec{v} - \vec{w} \right) = 0 $ (the algebraic condition derived in Theorem \ref{algconcom}). Then the indices corresponding to the nonzero components of $ \vec{v} $ and $ \vec{w} $ can be used to identify the sets  $ A $ and $ B $ (respectively) in the \texttt{EQUAL SUBSET SUM} problem.
	
	\m 
	
	\quad Suppose there is a solution $ H $ to \texttt{ILP-QCOMMUTE}. From Theorem~\ref{algconcom}, it follows that $ \vec{c}_{i} \cdot \left( \vec{v} - \vec{w} \right) = 0 $ for the vector $ \vec{c} $ associated with constraint $ C $ and any nonzero term in the basis $ \{ \mathbbm{1}, \sigma^{z}, \sigma^{+}, \sigma^{-} \}^{\otimes n} $ with a term $ \sigma^{\pm}$ on at least one qubit will be enough to define a new $ H' $ that will be associated with a solution to \texttt{EQUAL SUBSET SUM}. At least one such element exists for $ H $ because $ H $ has at least one off-diagonal term in the spin-z basis. We can associate any off-diagonal term of H with $ \vec{y}, \vec{v}, \vec{w} $ such that $ H' = \alpha \bigotimes_{i=1}^{n} \left( \sigma^{z} \right)^{y_{i}} \left( \sigma^{+} \right)^{v_{i}} \left( \sigma^{-} \right)^{w_{i}} + \alpha^{\dg} \bigotimes_{i=1}^{n} \left( \sigma^{z} \right)^{y_{i}} \left( \sigma^{+} \right)^{w_{i}} \left( \sigma^{-} \right)^{v_{i}} $ is a matrix with only that off-diagonal term and its complex conjugate for some $ \alpha $ and $ \vec{v} \neq \vec{0}, \vec{w} \neq \vec{0} $. Then for a specific off-diagonal term, every non-zero entry in $ \vec{v} $ between $ 1 $ and $ n $, call it $ i $, picks an integer $ s_{i} \in S $ for the set $ A $, and $ \vec{w} $ does likewise for the set $ B $, providing a solution to the corresponding instance of \texttt{EQUAL SUBSET SUM} since $ \vec{s} \cdot (\vec{v} - \vec{w}) = \left( \sum_{s_{i} \in A} s_{i} \right) - \left( \sum_{s_{i} \in B} s_{i} \right) = 0 $. Suppose there is a solution $ A, B $ to \texttt{EQUAL SUBSET SUM}, then define $ \vec{v} $ such that $ v_{i} = 1 $ ($ w_{i} = 1$) if and only if $ s_{i} $ is in $ A $ ($ B $). Then $ H = \bigotimes_{ i = 1 }^{ n } \left( \sigma^{+} \right)^{ v_{i} } \left( \sigma^{-} \right)^{ w_{i} }  + \bigotimes_{ i = 1 }^{ n } \left( \sigma^{+} \right)^{ w_{i} } \left( \sigma^{-} \right)^{ v_{i} } $ is a solution to \texttt{ILP-QCOMMUTE} since $ \left( \sum_{s_{i} \in A} s_{i} \right) + \left( \sum_{s_{i} \in B} s_{i} \right) = \vec{s} \cdot \left( \vec{v} - \vec{w} \right) = 0 $. Hence, we have the following result.

    \begin{thm}
        \texttt{ILP-QCOMMUTE} is NP-Hard.
    \end{thm}
    
    Given this result, we can show NP-Completeness by noting that we can check for an off-diagonal term in the spin-z basis in polynomial time. Let $ H $ be a proposed solution to the \texttt{ILP-QCOMMUTE} such that there exists nonequivalent indices $i,j$ such that entry $ h_{ij} \neq 0 $. Checking that $ H $ commutes with the constraints is polynomial time. For $ k \in \{ 1, \ldots, n \} $, check if $ \text{Tr}_{k}\left( H \sigma_{k}^{+} \right) \neq 0 $ or $ \text{Tr}_{k}\left( H \sigma_{k}^{-} \right) \neq 0 $. $ H $ has at least one element $h_{ij}$ that is off-diagonal in the spin-z basis if and only if there exists $ k $ such that at least one of these terms is nonzero. Since the partial trace of tensor products is the trace over a specific tensor component, this can be done quickly. 
    
    \begin{cor}
        \texttt{ILP-QCOMMUTE} is NP-Complete.
    \end{cor}
    
    While we have shown that this problem is NP-Hard, we note that for any specific instance of the problem, the practical runtime can still be tractable and therefore a useful avenue for even the hardest instances of optimization problems.Also take note that since every \textit{unitary} operator is the exponential of a corresponding \textit{Hermitian} operator, knowing the existence of a unitary operator that commutes with the constraint operators is paramount to knowing a Hermitian matrix exists with the same property. As such, our result immediately translates to the QAOA setting where one wishes to construct unitary operators that will \textit{commute} with the embedded linear constraint operators.
    
    \subsection{Bounded Weight ILP-QCOMMUTE}
    
    Despite the NP-hardness of \texttt{ILP-QCOMMUTE}, Section~\ref{sec:bounded} discusses a simple polynomial time algorithm to find driver terms up to some weight $ k $. Consider this modified version of \texttt{ILP-QCOMMUTE}, which asks about the existence of a Hermitian matrix that commutes with the constraints, but consists of interaction terms up to weight $ k $.
    
    \begin{defi}[ILP-QCOMMUTE-k-LOCAL] 
	Given a set $ \mathcal{C} = \{ C_{1}, \ldots, C_{m} \} $ of linear constraints such that $ \hat{C_{i}} = \sum_{j = 1}^{n} c_{ij} \sigma_{j}^{z} $, over a space $ \mathbb{C}^{2^{n}} $ with $ c_{ij} \in \Z $, is there a Hermitian Matrix $ H $, with $ \mathscr{O} \left( \text{poly}( n ) \right) $ nonzero coefficients over a basis $ \{ \chi_1, \chi_2, \chi_3, \chi_4 \}^{\bigotimes n} $ and no term with weight higher than $ k $, such that $ \left[ H, \hat{C_{i}} \right] = 0 $ for all $ \hat{C_{i}} $ and $ H $ has at least one off-diagonal term in the spin-z basis?
	\end{defi}
    
    \begin{thm}     
        \texttt{ILP-QCOMMUTE-k-LOCAL} is in \texttt{P} for k in $ \mathcal{O}(1) $.
    \end{thm}
    
    \begin{proof}
        Apply the brute force approach described in Section~\ref{sec:bounded}. Since $ k \in \mathcal{O}(1) $, the algorithm runs in time $ n^{\mathcal{O}(1)} $.
    \end{proof}
    
    This shows that for a practical application where the Hamiltonian driver should be all local, we can tractably find such a Hamiltonian driver. Moreover, any $ H $ that is all local and commutes with the constraints can be constructed by placing the right coefficients on the found terms by brute forcing the expression in Theorem~\ref{algconcom} (they form a basis for all Hamiltonians that commute with the constraints up to weight $ k $).
    
	\section{Reachability within the Feasible Space}\label{sec:reachability}
	
	\quad In the previous section, we proved that finding a Hermitian matrix which commutes with a collection of linear spin-z constraints is NP-Complete. The related question becomes finding a Hermitian matrix which commutes with the constraints, but also connects the feasibility space. Note that two states $ \ket{ p }, \ket{ q }$ are connected if they are in the same commutation subspace of $ H $, that is $ \bra{ p } H^{r} \ket{ q } \neq 0 $ for some $ r \in \Z^{+} $. As such, for any pair of solutions $i, j$ in the feasibility space, their associated vectors in the computational basis $ \ket{i}, \ket{j}$ should be in the same commutation subspace. In general, when finding a driver Hamiltonian for an anneal that should solve an optimization task, we wish to find a driver that satisfies this condition so that we can ensure that it will be able to reach the entire feasibility space, since commuting with the constraints alone is not enough to ensure this will happen. Consider the graph partitioning example discussed in Section~\ref{sec:background} and Section~\ref{sec:probdef}, clearly $\sigma_{1}^{+} \sigma_{2}^{-} + \sigma_{1}^{-} \sigma_{2}^{+} $ commutes with the constraint, but \textit{fails} to connect the entire feasible space. For example, the state $ \ket{0011} $ is \textit{disconnected} from $ \ket{1100} $ under this driver Hamiltonian. While it succeeds not to mix solution states to \textit{nonsolution} states, it does not mix all solution states with each other. We introduce the problem \texttt{ILP-QCOMMUTE-NONTRIVIAL} which asks to find a driver term that not only commutes with the constraints, but acts \textit{nontrivially} on the feasible space, so that the action of the driver is such that there exists one solution state which the driver maps to another solution state in the feasible space. If we think about the solution states as vertices in a graph, then transitions induced by driver terms are the edges (and a driver term can induce more than one such edge). Then the problem \texttt{ILP-QCOMMUTE-NONTRIVIAL} asks that the driver term found induces at least one edge in the graph of feasible solutions. For the graph partitioning example discussed above, the term we discussed is clearly a solution to the problem \texttt{ILP-QCOMMUTE-NONTRIVIAL} since it connects $ \ket{1010}$ to $ \ket{0110}$, both of which are in the feasible space for this example. Let $ P_{i}^{b_{i}} $ be the projection operator corresponding to the energy eigenvalue $ b_{i} $ for the constraint operator $ \hat{C}_{i} $. Then this formally defines the problem \texttt{ILP-QCOMMUTE-NONTRIVIAL}:
	
	\begin{defi}[ILP-QCOMMUTE-NONTRIVIAL]
	    Given a set $ \mathcal{C} = \{ C_{1}, \ldots, C_{m} \} $ of linear constraints and constraint values $ b = \{ b_{1}, \ldots, b_{m} \} $ such that $ \hat{C}_{i} = \sum_{j=1}^{n} c_{ij} \sigma_{j}^{z} $ over a space $ \mathbbm{C}^{2^{n}} $ with $ c_{ij} \in \Z $, is there a Hermitian Matrix H, with $ \mathcal{O}(poly(n)) $ nonzero coefficients over a basis $ \{ \chi_1, \chi_2, \chi_3, \chi_4 \}^{\bigotimes n } $, such that $ \left[ H, \hat{C}_{i} \right] = 0 $ for all $ \hat{C}_{i} $ and $ P_{1}^{ b_{1} } \, \cdots \, P_{ m }^{ b_{m} } H P_{ m }^{ b_{m} } \, \cdots \, P_{ 1 }^{ b_{1} } $ has at least one off-diagonal term in the spin-z basis?
	\end{defi}
	
	\quad The main difference is that while \texttt{ILP-QCOMMUTE} required nontrivial off-diagonal terms in the spin-z basis, \texttt{ILP-QCOMMUTE-NONTRIVIAL} specifically requires these to be nontrivial in the constraint space of interest, which in general is a non-polynomial problem to verify (i.e. knowing a Hamiltonian has or fails to have an eigenvector for a specific energy level would allow one to know if a Hamiltonian has a solution to hard problems). We show that this problem is at least NP-Hard by reducing a problem closely related to \texttt{EQUAL SUBSET SUM} to \texttt{ILP-QCOMMUTE-NONTRIVIAL}. We begin with the famous NP-Complete \texttt{SUBSET SUM} problem:
	
	\begin{defi}[SUBSET SUM]
	    Given a set $ S = \{ s_{1}, \ldots, s_{n} \} $ of integers and an integer target value $ T $, is there a subset $ S_1 $ such that $ \sum_{ s \in S_1 } s = T $?
	\end{defi}
	
	While \texttt{SUBSET SUM} asks about the existence of a single solution, we are interested in at least two solutions, defining the problem:
	
	\begin{defi}[2-OR-MORE SUBSET SUM]
	    Given a set $ S = \{ s_{1}, \ldots, s_{n} \} $ and a target value $ T $, are there two  subsets $ S_1, S_2 $ such that $ \sum_{ s \in S_1 } s = \sum_{ s \in S_2 } s = T $?
	\end{defi}
	
	We show that like \texttt{SUBSET SUM} (over positive integers\cite{cormen2009introduction}), \texttt{2-OR-MORE SUBSET SUM} is also NP-Hard:
	
	\begin{lem}
	    \texttt{2-OR-MORE SUBSET SUM} is NP-Hard.
	\end{lem}
	
	\begin{proof}
	
	    \quad Consider an instance of the \texttt{SUBSET SUM} problem with a set $ S = \{ s_{1}, \ldots, s_{n} \} $ and a target value $ T $ such that $ s_{i} > 0 $ for all $ s_{i} $. We construct a new instance of the \texttt{2-OR-MORE SUBSET SUM} problem with set $ S' = \{ s_{1}, \ldots, s_{n}, T \} $ and the same target value $ T $. 
	    
	    \quad First we show how to relate a solution to the original \texttt{SUBSET SUM} instance from a solution to the constructed \texttt{2-OR-MORE SUBSET SUM} problem. Let $ S_1, S_2 $ be solutions to the new \texttt{2-OR-MORE SUBSET SUM} instance, either one or none of the solutions uses the element $ T $. If neither does, either one is a solution to the instance of the original \texttt{SUBSET SUM} problem. Without loss of generality, suppose $ S_1 $ uses the value $ T $, then $ S_1 $ cannot use any other value since every other value is greater than zero, hence $ S_1 = \{ T \} $. Since $ S_2 $ cannot use any value other than $ T $ if $ T \in S_2 $ and $ S_1 \neq S_2 $, it follows that $ T \notin S_2 $. Then $ S_2 \subseteq S $ and $ \sum_{ s \in S_{2} } s = T $.
	    
	    \quad We now show how to relate a solution to the constructed \texttt{2-OR-MORE SUBSET SUM} instance given a solution to the \texttt{SUBSET SUM} instance. Let $ S_1 $ be a solution to the \texttt{SUBSET SUM} problem. Then $ S_1, \{ T \} $ is a solution to the \texttt{2-OR-MORE SUBSET SUM} instance. 
	    
	    \m
	    
	    \quad Since \texttt{SUBSET SUM} is NP-Hard over positive integers, \texttt{2-OR-MORE SUBSET SUM} is NP-Hard as well. 

	\end{proof}
	
    \quad \texttt{2-OR-MORE SUBSET SUM} is closely related to \texttt{EQUAL SUBSET SUM} because both ask about the existence of two subsets with equal sums, but \texttt{2-OR-MORE SUBSET SUM} adds the further constraint that these two subsets should have a specific sum. We show that \texttt{ILP-QCOMMUTE-NONTRIVIAL} is NP-Hard through a reduction to \texttt{2-OR-MORE SUBSET SUM}. Note again that \texttt{ILP-QCOMMUTE-NONTRIVIAL} is not verifiable in polynomial time \cite{kitaev2002classical}\cite{kempe20033}\cite{kempe2006complexity} and so this reduction is only for the decision version of the problem \texttt{2-OR-MORE SUBSET SUM}. 
	
	\begin{thm}
	    \texttt{ILP-QCOMMUTE-NONTRIVIAL} is NP-Hard.
	\end{thm}
	
	\begin{proof}
	    \quad Consider an instance of the \texttt{2-OR-MORE SUBSET SUM} problem with a set $ S = \{ s_{1}, \ldots, s_{n} \} $ and a target value $ T $. Define the constraint operator $ \hat{S} = \sum_{j=1}^{n} s_{j} \sigma_{j}^{z} $ and a target energy value $ \left(\sum_{j=1}^{n} s_{j}\right) - 2 \, T $.
	    
	    \m 
	    
	    \quad Suppose this instance of \texttt{ILP-QCOMMUTE-NONTRIVIAL} has a solution. Then there are at least two eigenvectors $ \ket{\vec{v}}, \ket{\vec{w}} $ of $ \hat{S} $ with eigenvalue $ \left( \sum_{j=1}^{n} s_{j} \right) - 2 \, T $ such that the two eigenvalues can be written in the spin-z basis with $ \vec{v}, \vec{w} \in \{ 0, 1 \}^{n} $. Like with \texttt{ILP-QCOMMUTE} and \texttt{EQUAL SUBSET SUM}, the nonzero elements of $ \vec{v} $ and $ \vec{w} $ describe two sets $ S_{1}, S_{2} $ such that $ s_{i} \in S_1 $ ($ s_{i} \in S_2 $) if and only if $ v_{i} = 1 $ ($ w_{i} = 1 $). Since $ \hat{S} \ket{ \vec{ v } } = \left( \sum_{j=1}^{n} s_{j} (1 - 2 \, v_j) \right) \ket{ \vec{v} } =  \left( \left( \sum_{j=1}^{n} s_{j} \right) - 2 \, \left( \sum_{j=1}^{n} \, s_{j} v_{j} \right) \right) \ket{ \vec{v} } $, it follows that $ \sum_{j=1}^{n} s_{j} v_{j} = T $. The same logic works for $ \vec{w} $ and so $ \sum_{ s \in S_1 } s = \sum_{ s \in S_{2} } s = T $. Then \texttt{2-OR-MORE SUBSET SUM} must have a solution as well, specifically $ S_1, S_2 $.
	    
	    \m 
	    
	    \quad Suppose \texttt{2-OR-MORE SUBSET SUM} has a solution. Then there are two nonequal subsets $ S_1, S_2 $ of $ S $ such that $ \sum_{ s_i \in S_1 } s_i = \sum_{ s_i \in S_2 } s_2 = T $. Then let $ \vec{v} = ( v_1, \ldots, v_n ) $ ($ \vec{w} = ( w_1, \ldots, w_n ) $) with $ v_i = 1 $ ($ w_i = 1 $) if $ s_i \in S_1 $ ($ s_i \in S_2 $). 
	    
	    Then $ \hat{S} \ket{ \vec{v} } = \left( \sum_{j=1}^{n} s_j (1 - 2 \, v_j) \right) \ket{ \vec{v} } = \left( \left( \sum_{j=1}^{n} s_{j} \right) - 2 \, \sum_{j=1}^{n} s_j v_j \right) \ket{ \vec{v} }= \left( \left( \sum_{j=1}^{n} s_{j} \right) - 2 \, T \right) \ket{ \vec{v} } $. The same logic works for $ \ket{ \vec{w} } $ and so $ \ket{ \vec{v} }, \ket{ \vec{w} } $ are both eigenvectors of $ \hat{S} $ with eigenvalue $ \left( \sum_{j=1}^{n} s_{j} \right) - 2 \, T $. Then $ \ket{ \vec{v} } \bra{ \vec{w} } + \ket{ \vec{w} } \bra{ \vec{v} } $ is a driver term that nontrivially maps solution states of this constraint problem to one another. 
	    
	\end{proof}
	
	\quad In practical applications, we are often able to quickly find some driver terms that commute with the constraints, but then need to know whether they are sufficient to connect the entire feasible space. This raises the following question: given $ k $ driver terms (individual basis terms) that commute with the constraints, does some linear combination of them with \textit{nonzero} coefficients connect the entire feasibility space? In other words, given that we have found $k$ driver terms that commute with the constraints, can we guarantee that some linear combination of them will have the whole feasible subspace as its smallest invariant subspace? Note that not every driver term that does commute with a subspace is necessary to solve this problem. For example, in the case of the constraint $ \hat{C} = \sum_{i}^{n} \sigma_{i}^{z} $, it suffices to use the driver terms $ \sigma_{i}^{+} \sigma_{i+1}^{-} + \sigma_{i}^{-} \sigma_{i+1}^{+} $ for $ i \in [ n - 1 ] $. Any linear combination with \textit{nonzero} coefficients, $ H_{d} = \sum_{i}^{n-1} \lambda_{i} \left( \sigma_{i}^{+} \sigma_{i+1}^{-} + \sigma_{i}^{-} \sigma_{i+1}^{+} \right) $, then is a valid Hamiltonian driver to connect the feasible space. Then an extra term, like $ \sigma_{1}^{+} \sigma_{3}^{-} + \sigma_{1}^{-} \sigma_{3}^{+} $ is unnecessary, because if $ \ket{\phi} $ is in the constrained subspace, it follows that $ ( \sigma_{1}^{+} \sigma_{2}^{-} + \sigma_{1}^{-} \sigma_{2}^{+} ) \ket{\phi} $, $ ( \sigma_{2}^{+} \sigma_{3}^{-} + \sigma_{2}^{-} \sigma_{3}^{+} ) \ket{\phi} $, $ ( \sigma_{1}^{+} \sigma_{2}^{-} + \sigma_{1}^{-} \sigma_{2}^{+} )( \sigma_{2}^{+} \sigma_{3}^{-} + \sigma_{2}^{-} \sigma_{3}^{+} ) \ket{\phi} $, and $ ( \sigma_{2}^{+} \sigma_{3}^{-} + \sigma_{2}^{-} \sigma_{3}^{+} )( \sigma_{1}^{+} \sigma_{2}^{-} + \sigma_{1}^{-} \sigma_{2}^{+} ) \ket{\phi} $ are as well. Note that $ ( \sigma_{1}^{+} \sigma_{3}^{-} + \sigma_{1}^{-} \sigma_{3}^{+} ) = ( \sigma_{1}^{+} \sigma_{2}^{-} + \sigma_{1}^{-} \sigma_{2}^{+} )( \sigma_{2}^{+} \sigma_{3}^{-} + \sigma_{2}^{-} \sigma_{3}^{+} ) + ( \sigma_{2}^{+} \sigma_{3}^{-} + \sigma_{2}^{-} \sigma_{3}^{+} )( \sigma_{1}^{+} \sigma_{2}^{-} + \sigma_{1}^{-} \sigma_{2}^{+} ) $. If a Hermitian matrix $ M $ can be decomposed into a linear combination of products of operators chosen from the set of driver terms $ \{ \hat{G}_{k} \}$, and $\ket{\phi}$ is a state in the constrained space, then for any state $\ket{\psi}$ such that $\bra{\psi} M \ket{\phi} \neq 0$ (i.e., any state reachable from $\ket{\phi}$ through the action of $M$), is also reachable through the action of the driver terms in $ \{ \hat{G}_{k} \}$ for the state $\ket{\psi}$.
	
	\m 
	
	\quad The set of Hamiltonians that commute with a given set of constraints form an algebra (known as the commutant of the set of constraints). Each one of the driver terms we are considering can be seen as a generator of this commutant algebra. This leads us to define 
	the problem \texttt{ILP-QIRREDUCIBLECOMMUTE-GIVEN-k} formally: 
	
	\begin{defi}[ILP-QIRREDUCIBLECOMMUTE-GIVEN-k]
	    Given a set $ \mathcal{C} = \{ C_{1}, \ldots, C_{m} \} $ of linear constraints and constraint values $ b = \{ b_{1}, \ldots, b_{m} \} $ such that $ \hat{C}_{i} = \sum_{j=1}^{n} c_{ij} \sigma_{j}^{z} $ over a space $ \mathbbm{C}^{2^{n}} $ with $ c_{ij} \in \Z $, and a set of basis terms $ \mathcal{G} = \{ \hat{G}_{1}, \ldots, \hat{G}_{k} \} $ such that $ \hat{G}_{i} \in \{ \chi_1, \chi_2, \chi_3, \chi_4 \}^{\bigotimes n} $, does $ \mathcal{G} $ connect the entire nonzero eigenspace of the operator $ P_{1}^{ b_{1} } \cdots P_{m}^{ b_{m} } $?
	\end{defi}
	
    \quad As such, \texttt{ILP-QIRREDUCIBLECOMMUTE-GIVEN-k} asks if a given set of driver terms is able to connect the entire feasible space of a set of constraints with the given constraint values. We show that this problem is also NP-Hard by reducing \texttt{ILP-QCOMMUTE-NONTRIVIAL} to \texttt{ILP-QIRREDUCIBLECOMMUTE-GIVEN-k}. We do so by finding a mapping for any instance of \texttt{ILP-QCOMMUTE-NONTRIVIAL} to an instance of \texttt{\justify ILP-QIRREDUCIBLECOMMUTE-GIVEN-k}. Consider such an instance with constraints $ \{ C_{1}, \ldots, C_{m} \} $ and constraint values $ \{ b_{1}, \ldots, b_{m} \} $. Find an integer $ a_{1} $ such that $ \| \vec{c_{i}} \|_{1} < a_{1} $ for $ i \in [m] $. Then expand the space $ \{ x_{1}, \ldots, x_{n} \} $ by appending $ x_{n + 1}, x_{n + 2} $. Make the constraints $ F_{i}(x) = C_{i}(x_{1}, \ldots, x_{n}) + a_{1} ( x_{n + 1} + x_{n + 2} ) $ with constraint values $ b_{i} + a_{1} $. Then we can easily find a driver term $ \hat{G}_{1} =  \sigma_{n + 1}^{+} \sigma_{n + 2}^{-} + \sigma_{n + 1}^{-} \sigma_{n + 2}^{+}$ (over the basis $ \{ \mathbbm{1}, \sigma^{+}, \sigma^{-}, \sigma^{z} \} $) such that $ [ \hat{G}_{1}, \hat{F_{i}} ] = 0 $ for $ i \in [ m ] $. Since the constraint values are $ b_{i} + a_{1} $, then the \textit{only} way to satisfy the constraints is if $ (x_{n+1}, x_{n+2}) \in \{(1,0),(0,1)\}$, since $ \sum_{i=1}^{n} c_{i} x_{i} \leq \| c_{i} \|_{1} < a_{1} $ by design.
	
	\m 
	
	\quad We have thus altered the constraints such that the feasibility space of the original problem is now doubled in size by the addition of the two variables $ x_{n+1}$ and  $x_{n+2}$, if the original feasibility space was non-empty. The important structure we note here is that the feasibility subspace of the new constraints is the tensor product of the feasibility subspace of the original constraints, with the subspace generated by the states $ \ket{10} $ and $ \ket{01} $ over qubits $ x_{n+1}, x_{n+2} $. It should also be easy to recognize then that the only nontrivial action induced by the driver terms we choose over $ x_{n+1}, x_{n+2} $ is precisely the action of $ \hat{G}_{1} $.
	
	\m
	
	\quad We can apply the same procedure recursively, adding ancillas $\{x_{n+1},\ldots,x_{n+2k}\}$ and generating $k$ driver terms, such that $ a_{i} > \| \vec{c}_{i} \|_{1} + \sum_{j=1}^{i-1} a_{j} $ and $ \hat{G}_{i} = (\sigma_{n+2\,i-1}^{+} \sigma_{n+2\,i}^{-} + \sigma_{n+2\,i-1}^{-} \sigma_{n+2\,i}^{+}) $, $1\leq i\leq k$. By this construction, when restricted to the ancilla variables, the feasible subspace is spanned by the vectors $\{ \bigotimes_{i=1}^k |i_1 i_2\rangle, i_1 + i_2 = 1\}$. Given any element in this subspace, the action of the $ \hat{G}_i$ driver terms is sufficient to guarantee that any other element of this subspace can also be generated. To proceed with the reduction, we then  give the constraints $ \{ F_{1} , \ldots, F_{m}  \} $ with constraint values $ \{ b_{1} + \sum_{i = 1}^{k} a_{i}, \ldots, b_{m} + \sum_{i = 1}^{k} a_{i} \} $ respectively and the drivers $ \{ \hat{G}_{1}, \ldots, \hat{G}_{k} \} $ to our \texttt{ILP-QIRREDUCIBLECOMMUTE-GIVEN-k} solver oracle. Since our $ k $ driver terms are all over the added qubits $ x_{n +1} $ to $ x_{n + 2 \, k} $, it should be clear that these driver terms say nothing about the feasibility space of the original problem over qubits $ x_{1} $ to $ x_{n} $. Suppose we are told that our $ k $ drivers are sufficient, i.e., they can generate the whole feasible subspace by acting on any one element of that subspace. Then clearly there are no driver terms for the original \texttt{ILP-QCOMMUTE-NONTRIVIAL} decision problem, since none of the $ k $ driver terms operate over qubits associated with variables $ x_{1} $ to $ x_{n} $. Likewise if we are told our $ k $ drivers are not sufficient, then clearly there must be \textit{at least one} nontrivial driver for the original problem, since the drivers are enough to generate all elements of the feasible subspace when restricted to the ancillas. Note that this solves \texttt{ILP-QCOMMUTE-NONTRIVIAL} without giving us a token to verify it. Because this is the most general unstructured version of the problem, is it possible that a different complexity result can be found for a more structured questioning of the same problem. We note that the problem can also have a stronger complexity result, such as a relationship to a higher class in the polynomial hierarchy\cite{meyer1972equivalence,stockmeyer1976polynomial,garey1979computers}, like \texttt{\#P}\cite{valiant1979complexity}, to which it has some natural analogues. 
	
	\m
	
	\quad Given this result and the result of Section~\ref{sec:bounded}, we can often find drivers that satisfy the condition stated in Theorem~\ref{algconcom}, but may not connect the entire feasible space. Still, there remain many avenues for exploiting such terms; for example, alongside the ordinary transverse field such that universality is maintained, but with biasing towards a subspace of the solution space. This gives us a way to adjust the knob of using higher order terms when the ordinary transverse field struggles to find a solution. Such driver terms can also be beneficial for exploring new solutions using reverse annealing\cite{venturelli2019reverse,king2018observation}, especially for solutions that are higher hamming distance away since the transverse field generally struggles to find such solutions.
	
	\m 
	
	\quad Another way to leverage our result is to brute force the problem for a set of constraints over a small enough subspace that it becomes polynomially tractable. Over the other variables, we apply the usual transverse field and enforce the other constraints as penalty terms in the final Hamiltonian. These approaches can also be adopted to the constraints that are \textit{geometrically} local (like in a two dimensional grid). 
	
	\section{Conclusion}\label{conclusion}
	
	\quad In this work, we addressed the computational complexity of finding driver Hamiltonians for quantum annealing processes which aim at solving optimization or feasibility problems with several linear constraints. We develop a simple and intuitive algebraic framework for understanding whether a Hamiltonian commutes with a set of constraints or not. While this result is interesting mathematically in its own right, we mainly focus on the problem posed in Ref.\cite{hen2016driver} about algorithmically finding driver Hamiltonians for optimization problems with several linear constraints. Most significantly, the condition is useful for finding a reduction of the NP-Hard problem \texttt{EQUAL SUBSET SUMS} to finding such a driver Hamiltonian, thereby allowing us to categorize the complexity of this problem. 
	
	\m 
	
	\quad We also showed that \texttt{ILP-QCOMMUTE-NONTRIVIAL} and \texttt{ILP-QIRREDUCIBLECOMMUTE-GIVEN-k} are at least NP-Hard. But these problems could well be in a higher complexity class in the polynomial hierarchy - like \texttt{\#P}, to which \texttt{\justify ILP-QIRREDUCIBLECOMMUTE-GIVEN-k} has some similarity. However, for most common implementations the Hamiltonians are of bounded weight, and the relevant complexity class \texttt{ILP-QCOMMUTE-k-LOCAL} for a small integer $ k $ is in \texttt{P}. Hence, there is a simple brute force algorithm, as detailed in Section~\ref{sec:bounded}, to find a basis for all possible driver Hamiltonians of this bounded locality. However, the results from \texttt{ILP-QIRREDUCIBLECOMMUTE-GIVEN-k} say that given a set of driver terms, it is intractable to know whether the found basis can sustain a Hamiltonian that connects the entire feasibility space for the linear constraints. As such, we present a polynomial time algorithm that is guaranteed to find a basis for all possible Hamiltonians that commute with a set of embedded constraint operators up to a certain weight, but with no guarantees that the found Hamiltonian is able to connect the entire feasibility space that the constraints specified. However, for some important problems it is actually possible to exploit the constraint structure to guarantee that driver terms with low weight will be sufficient to reach all feasible states. This is the case, for example, for the graph coloring problem discussed in Section~\ref{sec:background} and Ref.\cite{hen2016driver}.
	
	\m
	
	\quad Our result  also applies to finding \textit{mixing} operators for Quantum Alternating Operator Ansatz (QAOA)\cite{farhi2014quantum,hadfield2019quantum}. To implement highly nontrivial driver Hamiltonians for an anneal, it also becomes necessary to find a new initial Hamiltonian that is then evolved slowly with a simple linear interpolation to the driver Hamiltonian since thermal equilibration to the driver Hamiltonian may be difficult. It then becomes relevant how  can we construct such a Hamiltonian for a given driver Hamiltonian, such that we can guarantee that we reach the right constrained space. This is a fundamental question for future research. While we have shown these problems to be NP-Hard, we have not shown what the average hardness of this class is or what the typical hardness is for instances of interest for specific applications. Especially pertinent become sets of instances in which the practical runtime for finding driver Hamiltonians remains tractable or the hardness of the problem comes from having to search a large feasible space for an optimum solution rather than pinpointing a very small feasible space. It is also interesting to note that our algebraic formulation is agnostic to the stoquasticity of the terms found. In the presented basis of Section~\ref{algcondition}, the stoquasticity of the individual basis terms, as written in Eq.~\ref{hexpand}, is determined by the amplitudes $ \alpha $ and its conjugate pair $ \alpha^{\dg} $. Commutation is \textit{invariant} under altering $ \alpha $ ($ \alpha^{\dg} $ will adjust as we alter $ \alpha $ to keep the term pair commutative). Once we have found driver terms that are suitable for a problem, it then raises the question of what effect, if any, choosing coefficients that will make them stoquastic or non-stoquastic will have on the anneal\cite{bravyi2006complexity,bravyi2010complexity,marvian2019computational,crosson2020signing}.This is another direction that requires further study. 
	
	
	\section{Acknowledgments}
	The research is based upon work (partially) supported by the Office of
the Director of National Intelligence (ODNI), Intelligence Advanced
Research Projects Activity (IARPA) and the Defense Advanced Research Projects Agency (DARPA), via the U.S. Army Research Office
contract W911NF-17-C-0050. The views and conclusions contained herein are
those of the authors and should not be interpreted as necessarily
representing the official policies or endorsements, either expressed or
implied, of the ODNI, IARPA, or the U.S. Government. The U.S. Government
is authorized to reproduce and distribute reprints for Governmental
purposes notwithstanding any copyright annotation thereon.
	
	\pagebreak

	\bibliographystyle{IEEEtran}
	\bibliography{main}

	\appendix
    \section{0-1-LP-QCOMMUTE is NP HARD}\label{sec:redoxold}
	
	\quad We reduce the \texttt{0-1-LP-QCOMMUTE} problem to the \texttt{EQUAL SUBSET SUM} problem. We define the \texttt{EQUAL SUBSET SUM} problem as before:
	
	\begin{defi}{EQUAL SUBSET SUM}
		Given a set $ S = \{ s_{1}, s_{2}, \ldots, s_{n} \} $, with $ s_{i} \in \Z^{+} $, find two non-empty disjoint subsets, $ A, B $ such that $ \sum_{a_{i} \in A} a_{i} = \sum_{b_{i} \in B} b_{i} $.
	\end{defi}
	
	\quad The \texttt{EQUAL SUBSET SUM} problem is known to be NP-Complete\cite{woeginger1992equal}. We map an instance of the \texttt{EQUAL SUBSET SUM} problem to the \texttt{0-1-LP-QCOMMUTE} problem; the former of which is defined over a set $ S = \{ s_{1}, s_{2}, \ldots, s_{n} \} $, with $ s_{i} \in \Z^{+} $. In order to connect \texttt{EQUAL SUBSET SUM} with solving a linear system over discrete variables (the key of Theorem \ref{algconcom}), we will associate an assignment of integers in $S$ to the two subsets $A$ and $B$ with a function $ u $  over $ S $ such that $ u = \{ u_{1}, \ldots, u_{n} \} $ with  $ u_{i} \in \{ -1, 0, 1 \} $. We associate the value $ u_{i} $ as the assignment $ u $ gives to integer $ s_{i} $. Slightly abusing notation, this defines a function on any subset $ M = \{ s_{m_{1}}, s_{m_{2}}, \ldots, s_{m_{|M|}} \} $ such that $ u(M) = \{ u_{m_{1}}, u_{m_{2}}, \ldots, u_{m_{|M|}} \} $. When discussing a subset of a single element $ s_{e} $, we also abuse notation to allow for $ u(s_{e}) = u_{e} $. We can then define an integer valued function $ E_{S} (u) = \sum_{s_{i} \in S}^{n} u_{i} s_{i} $. If we associate integers $s_i$ such that $u_i = 1$ with integers in subset $A$, and those with $u_i = -1$ with integers in subset $B$, then we can rewrite it as $E_{S} (u) = \sum_{s_i \in A} s_i - \sum_{s_i \in B} s_i $ (note that $u_i = 0$ means that the corresponding integer is not chosen for any of the two subsets). Then, \texttt{EQUAL SUBSET SUM} has a solution if and only if there is an assignment function $u$ with a nontrivial image such that $ E_{S} (u) = \sum_{i}^{n} u_{i} s_{i} = 0 $.
    
    \m 

    \quad However, we \textit{need} a vector representation to exploit the structure of Theorem \ref{algconcom}. Let $ s_{\text{max}} $ be the maximum of $ S $. We define $ S^{M} $ as the matrix with the binary representations of $ S $ as its column vectors. Given $ s_{j} \in S $, we define entry $ s_{ij}^{M} = s_{j}^{i} $  -- referring to the $ i $-th bit of integer $ s_{j} $. This defines a $ m \times n $ matrix with $ m = \lceil \log( s_{max} ) \rceil $. 
    
    \m 
    
    \quad The idea is that we wish to give each integer an associated binary vector such that multiplying a binary vector with $ S^{M} $ corresponds to \textit{selecting} that integer to participate in a sum. We refer to the vectorized form of $ u $ as $ \vec{\mu} \in \{ -1, 0, 1 \}^{n} $ such that $ \vec{\mu} = ( u_{1}, \ldots, u_{n} ) $. Since multiplying a matrix by a vector on the right results in a linear combination of the matrix columns, with the coefficients being the corresponding components of the vector, it would be tempting to assume that $ E_{S} (u) =  S^{M} \, \vec{\mu}$, since the columns of $ S^{M} $ are associated with the integers in $S$. Then we would have something like $ E_{S} (u) = 0$ if and only if $S^{M} \, \vec{\mu} = \vec{0}$, providing our desired connection between Theorem \ref{algconcom} and \texttt{EQUAL SUBSET SUM}.
    
    \m 
    
    \quad Unfortunately this does not work, since the columns of $S$ contain a binary representations of the integers $s_i$, while the expression $ E_{S}(u) $ refers to the usual addition of integers, and not bit component wise addition. To illustrate what we mean with this, consider the (improper) set $  S = \{ 1, 1, 2 \} $ which delineates:
    \begin{align*}
        S^{M} = \; \begin{array}{ccccc}
                    & s_1 & s_2 & s_3 &     \\
\ldelim({2}{0.5em}  & 1 & 1 & 0 & \hspace{-0.2em}\rdelim){2}{0.5em}  \\
                    & 0 & 0 & 1 &     \\
                    \end{array}
	\end{align*}
    
    Even though the associated \texttt{EQUAL SUBSET SUM} problem has a simple solution associated with the function $ u = \{ 1, 1, -1 \} $ (assign the first two integers to subset $A$ and the third to subset $B$), a simple calculation shows that $ S^{M} \, \vec{\mu} = (2, -1) \neq \vec{0}$. From the example above we can see that what we are missing is a way of incorporating the ``bit carry" that occurs in binary addition into the operations of regular matrix-vector multiplication.  The main goal of this appendix is to show how this can be accomplished by embedding these matrix operations into a larger vector space.
	
    \m
	
	
	\m 
	
	\quad In order to resolve this issue, we will introduce a mechanism to do \textit{generalized} bit addition - bit addition that is generalized to when the bit values can both be positive and negative as well as zero. We add ancillary bits $ \mathscr{A} $ such that $ u^{\ast} $ is the assignment $ u $ expanded to this new space $ S \cup \mathscr{A} $ as $ u^{\ast} = \{ u_{1}, \ldots, u_{n}, u_{n+1} ,\ldots, u_{n+|\mathscr{A}|} \} $. Slightly abusing notation, for any subset $ M = M_{S} \cup M_{\mathscr{A}} $ with $ M_{S} = \{ s_{ms_{1}}, \ldots, s_{ms_{|M_{S}|}} \} $ and $ M_{\mathscr{A}} = \{ a_{ma_{1}}, \ldots, a_{ma_{|M_{\mathscr{A}}|}} \} $, we define $ u^{\ast}(M) = \{ u_{ms_{1}}, \ldots, u_{ms_{|M_{S}|}}, a_{ma_{1}}, \ldots, a_{ma_{|M_{\mathscr{A}}|}} \} $. We construct new constraints $ \mathscr{K} $ such that $ E_{S}(u) = 0 \mim E_{\mathscr{K}}(u^{\ast}) = 0 $. Moreover, $ u^{\ast} $ will allow for a vectorized form $ \vec{\mu^{\ast}} $ and $ \mathscr{K} $ a matrix $ K^{M} $ (see \ref{sec:simplereduced}) such that $ E_{\mathscr{K}}(u^{\ast}) = 0 \mim K^{M} \, \vec{\mu^{\ast}} = 0 $. Intuitively, $ u^{\ast} $ \textit{picks} coefficients for values over $ S $ and is subsequently \textit{forced} to take values on $ \mathscr{A} $ corresponding to doing valid bit addition and only satisfies $ \mathscr{K} $ if the bit entries of the total sum is indeed zero. Fig.~\ref{fig:reduxflowchart} gives a visual description of the steps used to create our full reduction. As such, for a given set of integers $ S $, we follow the reduction to construct a \textit{binary matrix} $ K^{M} $ such that the row vectors of $ K^{M} $ define the constraint operators $ \hat{K}_{i} = \sum_{j=1}^{|S \cup \mathscr{A}|} k_{ij}^{M} \sigma_{j}^{z} $. This serves as the input binary LP to the oracle solver of \texttt{0-1-LP-QCOMMUTE} to tell us if a Hamiltonian $ H $ exists such that $ H $ has an off-diagonal term in the spin-z basis that shows the existence of $ \vec{v}, \vec{w} $, which describe two subsets $ A $ and $ B $ as the solution to the given \texttt{Equal Subset Sum} problem.
	
	\begin{figure}
	\centering
 \makebox[\textwidth][c]{\includegraphics[height=0.4\textheight]{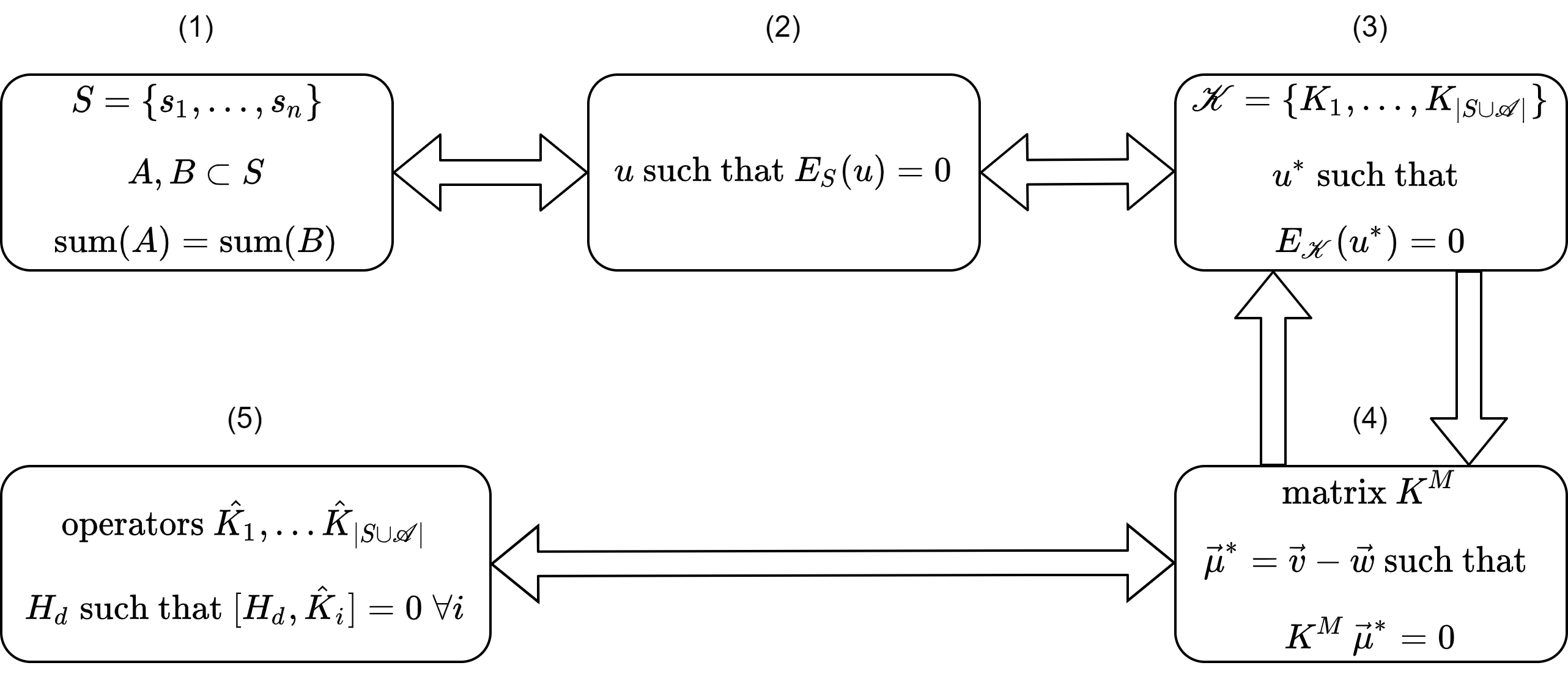}}
	\caption{A flow cart describing how our reduction works, we recommend motivated readers refer back to it as they read the reduction. An instance of \texttt{EQUAL SUBSET SUM} (box 1) is mapped into a binary constraint representation such that the sum function $ E $ defined over the assignment $ u $ is equivalent to $\text{sum}(A) - \text{sum}(B)$ where $ u $ assigns variables to either $ A $, $ B $, or they are not used (box 2). To exploit Theorem~\ref{algconcom}, constraints $ C $ are mapped to constraints $ \mathscr{K} $ (box 3), such that assignment $ E_{\mathscr{K}}(u^{\ast}) = 0 \mim E_{S}(u^{\ast}) = 0 $. Unlike $ S $, $ \mathscr{K} $ allows for a simple matrix representation such that $ K^{M} \, \vec{\mu^{\ast}} = 0 \mim E_{\mathscr{K}}(u^{\ast}) = 0 $ (box 4), where $ \vec{\mu} $ is a naive vectorized form of $ u^{\ast} $. Note that if u exists such that $ E_{S}(u) = 0 $, then many $ u^{\ast} $ exist such that $ E_{\mathscr{K}}(u^{\ast}) = 0$, but each reduces to the \textit{same} $ u $. The constraint version of $ K^{M} $ can be embedded row-wise to define operators $ \hat{K}_{1}, \ldots, \hat{K_{S \cup \mathscr{A}}} $ as $ \hat{K}_{i} = \sum_{j=1}^{|S\cup\mathscr{A}|} k_{ij}^{M} \sigma_{j}^{z} $ such that a \texttt{0-1-LP-QCOMMUTE} oracle solves to show the existence of a driver Hamiltonian $ H_{d} $, which we can interpret back to see there must be a solution to \texttt{EQUAL SUBSET SUM} as well.}
 	\label{fig:reduxflowchart}
	\end{figure}
	
	\subsection{Generalized Full Adder}\label{subsec:genadd}
	
	\begin{table}
		\centering
		\begin{tabular}{ c  c | c  c | c c }
			\multicolumn{2}{c}{Inputs} & \multicolumn{4}{c}{Output}	\\
			\hline
			&		&	\multicolumn{2}{c}{Primary} & \multicolumn{2}{c}{Secondary} 	\\
			\hline
			a   &   b   &   s   &   c   &	s	&	c	\\
			\hline 
			-1  	&  -1   &   0   &  -1   &	-	&	-	\\
			-1  	&   0   &  -1   &   0   &	1	&	-1	\\
			-1  	&   1   &   0   &   0   &	-	&	-	\\
			0   &  -1   &  -1   &   0   &	1	&	-1	\\
			0   &   0   &   0   &   0   &	-	&	-	\\
			0   &   1   &   1   &   0   &	-1	&	1	\\
			1   &  -1   &   0   &   0   &	-	&	-	\\
			1   &   0   &   1   &   0   &	-1	&	1	\\
			1   &   1   &   0   &   1   &	-	&	-	\\
		\end{tabular}
		\caption{The generalized full adder; if $ u^{\ast} $ takes a particular value on inputs $ a $ and $ b $, then $ u^{\ast} $ will  be forced to take the corresponding sum (represented by $ s $) and carry (represented by $ c $) values. In the case that $ a + b $ is not a power of two, $ s $ and $ c $ have two possible values they can take. Here primary (secondary) operations correspond to the operations where the carry is set to zero (nonzero) if possible. }
		\label{tab:genadd}
	\end{table}
	
	\quad In this section we describe how to build the basis for our reduction, which is to find a matrix such that the values $ u^{\ast} $ takes on the set $ S $ are added bitwise over the ancillary bits $ \mathscr{A} $. There will be specific ancillary bits such that the total sum that $ u^{\ast} $ takes on $ S $ can be deduced from its value on these bits. Consider again the simple example we introduced in the previous section. We will add ancillary variables such that their values are \textit{forced} to be what is dictated by the bit addition of values in $ S $. This can be summarized in Table \ref{tab:genadd}. If $ u^{\ast} $ takes a particular value on two inputs $ a $ and $ b $, then the table describes what value $ u^{\ast} $ will be \textit{forced} to take on new ancillary values $ s $ and $ c $ (representing the sum and carry bits respectively).
	
	\m 
	
	\quad Like the ordinary adder, the generalized adder accepts all values such that $ u^{\ast}( a ) + u^{\ast}( b ) = 2 \, u^{\ast}( c ) + u^{\ast}( s ) $ except now $ u^{\ast}( x ) \in \{ -1, 0, 1 \} $ for any $ x $ and so $ u^{\ast}( a ) + u^{\ast}( b ) \in \{ -2, -1, 0, 1, 2 \} $. Note that then the carry bit and the sum bit are not unique like in the case of the ordinary full adder.	For example if $ u^{\ast}( a ) = 1 $ and $ u^{\ast}( b ) = 0 $, then it is possible that $ u^{\ast}( c ) = 0 $ and $ u^{\ast}( s ) = 1 $ like in the ordinary adder, but \textit{also} that $ u^{\ast}( c ) = 1 $ and $ u^{\ast}( s ) = -1 $. Since $ 2 \, u^{\ast}( c ) + u^{\ast}( s ) $ is the same value for both, they are both technically valid. The operations keen to the ordinary full adder we refer to as \textit{primary} and those that do not as \textit{secondary}. When possible, a primary operation will set the carry bit to zero while a secondary operation will set the carry bit to either one or negative one. One may hope that we could \textit{force} the \textit{primary} mode of operation, but we could not construct a 0-1 matrix that could force these modes of operations over the \textit{secondary} modes since our condition for satisfaction is through equivalence statements like $ u^{\ast}( a ) + u^{\ast}( b ) = 2 \, u^{\ast}( c ) + u^{\ast}( s ) $, but no equivalence statement can state a preference in representation. While it does not affect the correctness of our result, it does mean that the number of solutions is not preserved in our reduction - there are many valid $ u^{\ast} $ that reduced to a single $ u $. The reduction is therefore not parsimonious. 
	
	\m
	
	\quad To enforce the generalized adder between two inputs and two outputs we need to generate the correct submatrix. Given inputs $ a $ and $ b $, we define the matrix on $ a, b, s, c, x_{1}, x_{2}, x_{3} $ - with $ x_{1}, x_{2}, x_{3} $ being intermediating ancillas - as:
	\begin{align}
        GA^{M} = \; \begin{array}{c c c c c c c c c}
                    &   a & b & x_1 & x_2 & x_3 & c & s &     \\ 
\ldelim({6}{0.5em}  &   a  &  b  &   1  &   1   &   1   &   0   &   0   & \hspace{-0.2em}\rdelim){6}{0.5em}  \\
	                &   0  &  0  &   1  &   0   &   1   &   1   &   1   &	\\
	                &   0  &  0  &   0  &   1   &   1   &   1   &   1   &	\\
	                &   0  &  0  &   1  &   0   &   0   &   1   &   0   &	\\
	                &   0  &  0  &   0  &   1   &   0   &   1   &   0	&	\\
	                &   0  &  0  &   0  &   0   &   1   &   0   &   1	&	
                    \end{array}
	\label{gamat}
	\end{align}

	As constraints, we can write it as:
	\begin{align}
	GA_{1}( a, b, x_{1}, x_{2}, x_{3})  &= 0, 	\\
	GA_{2}( x_{1}, x_{3}, s, c)         &= 0, 	\\
	GA_{3}( x_{2}, x_{3}, s, c )        &= 0, 	\\
	GA_{4}( x_{1}, c )                  &= 0, 	\\
	GA_{5}( x_{2}, c )                  &= 0, 	\\
	GA_{6}( x_{3}, s )                  &= 0.
	\end{align}
	
	\quad For every generalized adder in Fig.~\ref{fig:fullcircuit} (as described in the protocol we gave in Section~\ref{subsec:genadd}), we have a submatrix over the corresponding variables. We give a simple case by case proof that $ GA^{M} $ enforces $ u^{\ast} $ to be valid if and only if its entries satisfy $ 2 \, u^{\ast}( c ) + u^{\ast}( s ) = u^{\ast}( a ) + u^{\ast}( b ) $ as seen in Fig.~\ref{tab:genadd} in Appendix~\ref{proofmatrix}. 
	
	\subsection{The Simple Reduced Case}\label{sec:simplereduced}
	
	\quad Before we move on to give a general protocol for any given problem, we consider the simple case we described earlier with the integer (improper) set $ S = \{ 1, 1, 2 \} $. We give a slightly reduced description for this problem to show what the reductions typically look like. We implement a generalized adder for the bits $ s_{1}^{1} $ and $ s_{2}^{1} $  - introducing the ancillary bits $ k_{1}^{1}, z_{1}^{1} $ that are the corresponding carry and sum bit. We then implement a generalized adder for the bits $ s_{3}^{2} $ and $ k_{1}^{1} $ - introducing the ancillary bits $ k_{2}^{1}, z_{2}^{1} $ that the corresponding carry and sum bit. As such, $ E_{S}(u^{\ast}) = 0 \mim u^{\ast}( z_{1}^{1} ) = u^{\ast}( z_{2}^{1} ) = u^{\ast}( k_{2}^{1} ) = 0 $, since the latter condition is equivalent to saying that the bitwise sum of the two sets is zero. This is represented in Fig.~\ref{fig:examplecircuit}A. Each box in the diagram refers to a generalized full adder. In words, we add the assignments $ u_{1}^{\ast}, u_{2}^{\ast} $ of the bits $ s_{1}^{1}, s_{2}^{1} $ and add the respective carry bit assignment with the assignment $ u_{3}^{\ast} $ on the bit $ s_{3}^{2} $. The resulting integer is given by $ E_{\{z_1^1, z_2^1, k_2^1\}}\left( u^{\ast}( \{z_1^1, z_2^1, k_2^1\} ) \right) = u^{\ast}( z_1^1 ) + 2 \times u^{\ast}( z_2^1 ) + 4 \times u^{\ast}( k_2^1 ) $ - the first row sum bit, the second row sum bit, and what can be considered the third row sum bit added with their respective power of two. This must be \textit{zero} if $ u^{\ast} $ defines subsets of equal sums and therefore $ u^{\ast} $ must be zero on each of them.
	
	\begin{figure}
		\begin{tabular}{ c c } 
			\includegraphics[height=0.3\textheight]{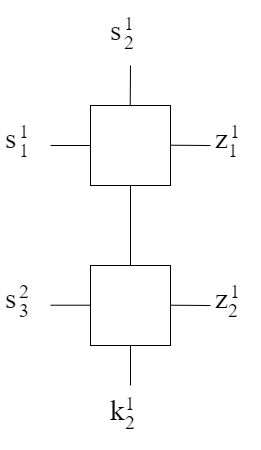} & \includegraphics[height=0.3\textheight]{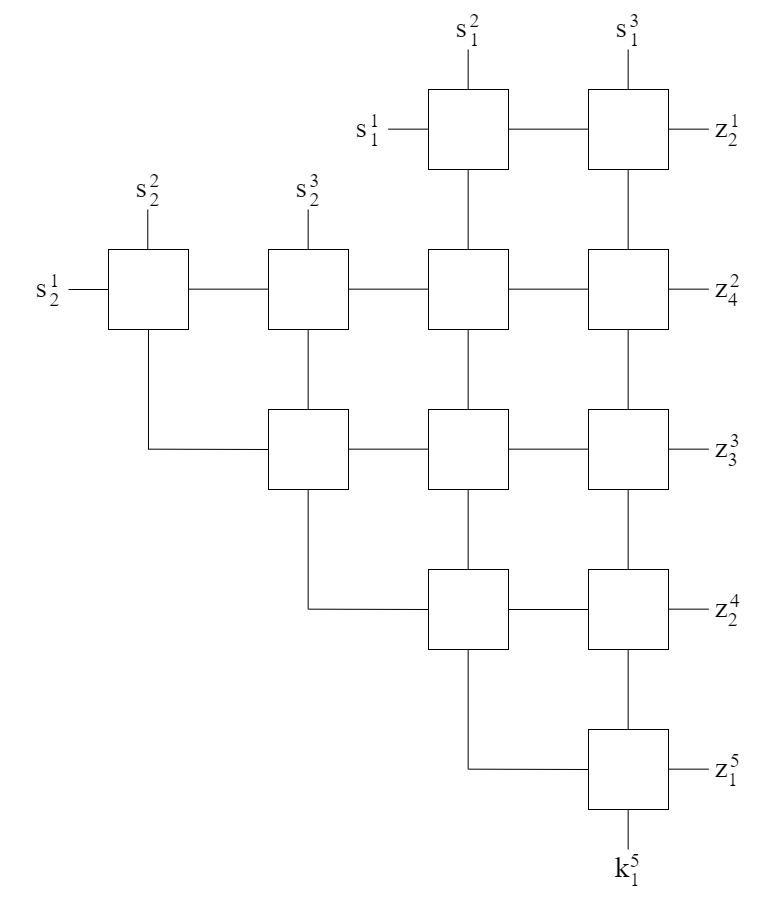} \\
			A	& 	B
		\end{tabular}
		\caption{Subfigure A shows the reduced embedding of the \texttt{EQUAL SUBSET SUM} instance with the (improper) integer set $ \{ 1, 1, 2 \} $. Each box represents a generalized full adder. Each adder describes a corresponding submatrix in the matrix $ \tilde{K}^{M} $ (check Eq.~\ref{eq:kexample}). Subfigure B shows the full embedding of the same instance.}
		\label{fig:examplecircuit}
	\end{figure}
	
	\m 
	
	\quad The resulting matrix $ \tilde{K}^{M} $ - here we use a tilde to signify that we are in the reduced construction case - that this process defines can be represented as: 
	\begin{align}\label{eq:kexample}
    \tilde{K}^{M} = \; \begin{array}{ c c c c | c c c | c c | c c c | c c c }
                    & s_1^1 & s_2^1 & s_3^2 & x_1^1 & x_1^2 & x_1^3 & k_1^1 & z_1^1 & x_2^1 & x_2^2 & x_2^3 & k_2^1 & z_2^1 &      \\
\ldelim({15}{0.5em} &   1   	&   1   	&   0   	&   1   	&   1   	&   1   	&   0   	&   0   	&   0   	&   0   	&	0		&   0   	&   0   & \hspace{-0.2em}\rdelim){15}{0.5em}  		\\
&   0   	&   0   	&   0   	&   1   	&   0   	&   1   	&   1   	&   1   	&   0   	&   0   	&	0		&   0   	&   0   &			\\
&   0   	&   0   	&   0   	&   0   	&   1   	&   1   	&   1   	&   1   	&   0   	&   0   	&	0		&   0   	&   0   &			\\
&   0   	&   0   	&   0   	&   1  		&   0   	&   0   	&   1   	&   0   	&   0   	&   0   	&	0		&   0   	&   0   &			\\
&   0   	&   0   	&   0   	&   0   	&   1   	&   0   	&   1   	&   0		&   0   	&   0   	&	0		&   0   	&   0   & 			\\
&   0   	&   0   	&   0   	&   0  		&   0   	&   1   	&   0   	&   1		&   0   	&   0   	&	0		&   0   	&   0   &			\\ 
	\cline{2-14}
&   0   	&   0   	&   1   	&   0  		&   0   	&   0   	&   1   	&   0		&   1   	&   1   	&	1		&   0   	&   0   & 			\\
&   0   	&   0   	&   0   	&   0  		&   0   	&   0   	&   0   	&   0		&   1 		&   0   	&	1		&   1   	&   1   & 			\\
&   0   	&   0   	&   0   	&   0  		&   0   	&   0   	&   0   	&   0		&   0 		&   1   	&	1		&   1   	&   1   & 			\\
&   0   	&   0   	&   0   	&   0  		&   0   	&   0   	&   0   	&   0		&   1 		&   0   	&	0		&   1   	&   0   & 			\\
&   0   	&   0   	&   0   	&   0  		&   0   	&   0   	&   0   	&   0		&   0 		&   1   	&	0		&   1   	&   0   & 			\\
&   0   	&   0   	&   0   	&   0  		&   0   	&   0   	&   0   	&   0		&   0 		&   0   	&	1		&   0   	&   1   & 			\\
	\cline{2-14}
&   0   	&   0   	&   0   	&   0  		&   0   	&   0   	&   0   	&   1		&   0 		&   0   	&	0		&   0   	&   0   & 			\\
&   0   	&   0   	&   0   	&   0  		&   0   	&   0   	&   0   	&   0		&   0 		&   0   	&	0		&   1   	&   0   & 			\\
&   0   	&   0   	&   0   	&   0  		&   0   	&   0   	&   0   	&   0		&   0 		&   0   	&	0		&   0   	&   1   & 		
                \end{array}
	\end{align}
	
	\quad One can check that if $ \vec{ \mu^{\ast} } = ( 1, 1, -1, -1, -1, 0, 1, 0, 0, 0, 0, 0, 0 ) $, then $ \tilde{K}^{M} \, \vec{ \mu^{\ast} } = \vec{ 0 } $. The vector $ \vec{ \mu^{\ast} } $ defines the assignment $ u^{\ast}( S ) = \{ 1, 1, -1 \} $ - since $ s_{1}, s_{2}, s_{3} $ are the first three entries of the vectorized form. This defines the two sets $ A = \{ s_{1}, s_{2} \} $ and $ B = \{ s_{3} \} $ as a solution to the \texttt{EQUAL SUBSET SUM} problem posed. One can check that $ \vec{ \mu^{\ast} } = (1,-1, 0, 0, 0, 0, 0, 0, 0, 0, 0, 0, 0) $ is also a solution, corresponding to the sets $ A = \{ 1 \} $ and $ B = \{ 1 \} $. 
	
	\m 
	
	\quad In this reduced construction, we only used the generalized full adder for the significant bits of each $ s_{j} \in S $ for a given bit entry $ i $. This helps to greatly reduce the size of the resulting embedding, but hopefully still conveys the principal idea behind our reduction. While we could write a general protocol on the same principle, it requires a more involved strategy than the one we take.
	
	\subsection{The Simple Unreduced Case}
	
	\quad To simplify the construction of the embedding at the cost of increasing their corresponding size, we follow the same logic as before, but do not \textit{prune} the insignificant bits. In the unreduced construction, we ``compute'' the sum bit by bit. Like in the reduced case, the resulting sum bit at the end of each layer corresponds to a bit entry that the sum defined by $ u^{\ast} $ takes - remember, in the end, the value of $ \sum_{a \in A} a - \sum_{b \in B} b $ was described bitwise by the value $ u^{\ast} $ took on the last sum bit in each layer plus the last layer's carry bit (e.g. $ E_{S}(u^{\ast}(S)) = u^{\ast}( z_{1} ) + 2 \times u^{\ast}( z_{2} ) + 4 \times u^{\ast}( k_{2} ) $). This remains the same in the unreduced representation. Note that nonzero sums can have multiple bit representations when the entries can be negative or positive, i.e. $ 1 = -1 \times 1 + 1 \times 2 = 1 \times 1 + 0 \times 2 $, while the sum $ 0 $ has only one. 
	
	\m
	
	\quad In words, we add the bits of each integer in the corresponding bit entry as well as all carry bits from the previous layer to find the total sum of all the integers if none of them had significant bits beyond this layer. Let $ u^{\ast}( z_{end}^{i} ) $ be the last sum bit for any row $ i $. Then the total sum up to the current layer $ i $ can be written as $ 2^{i} \left( 2 \times q + s \right) + \sum_{j=1}^{i} 2^{j-1} u^{\ast}( z_{end}^{j} ) $ for some $ q $. Then we identify $ s $ as the last sum bit of the current layer, $ u^{\ast}( z_{end}^{i} ) $ and $ q $ as the \textit{net} number of carries passed from layer $ i $ to $ i + 1 $.
	
	\m 
	
	\quad Consider again the simple case we described earlier. We have a (improper) set $ S = \{ 1, 1, 2 \} $, such that we can identify each of these three values as $ s_{1}, s_{2}, s_{3} $. Refer to Fig.~\ref{fig:examplecircuit}B to see the resulting diagram that this construction will give. Then $ s_{1}^{1}, s_{2}^{1}, s_{3}^{1} $ are the \textit{first} bits of each of these three. We use generalized full adders to add - bit by bit - the values $ s_{1}^{1}, s_{2}^{1}, s_{3}^{1} $ and feed the resulting carry bits $ k_{1}^{1}, k_{2}^{1} $ to the next layer while $ z_{2}^{1} $ takes the value of the \textit{lowest bit entry} for the total sum of the assignment. In the \textit{second} layer we add - bit by bit - the values $ s_{1}^{2}, s_{2}^{2}, s_{3}^{3}, k_{1}^{1}, k_{2}^{1} $ and feed the resulting carry bits $ k_{1}^{2}, k_{2}^{2}, k_{3}^{2}, k_{4}^{2} $ to the next layer while $ z_{2}^{4} $ takes the value of the \textit{second} lowest bit entry for the total sum. Since the maximum bit entry was given in row two, row \textit{three} adds \textit{only} the carry bits $ k_{1}^{2}, k_{2}^{2}, k_{3}^{2}, k_{4}^{2} $, which generates the carry bits $ k_{1}^{3}, k_{2}^{3}, k_{3}^{3} $ that are subsequently fed into layer \textit{four} while $ z_{4}^{3} $ is the \textit{third} lowest bit entry for the total sum. Layer \textit{four} adds the carry values and generates the corresponding carry bits $ k_{1}^{4}, k_{2}^{4} $ as well as the sum bit $ z_{2}^{4} $. Lastly, layer five adds these values and generates the corresponding carry bit $ k_{1}^{5} $ as well as the last sum bit $ z_{1}^{4} $. To complete our description, each layer has internal sum variables from each generalized adder. Every line in the diagram corresponds to a variable; variables that are between two boxes are intermediaries, such as all the carry bits except $ k_{1}^{5} $ and all the sum bits except the last ones in each layer. For example, layer one has $ z_{1}^{1} $ - an intermediate sum bit that is passed from the first generalized adder to the \textit{second}. This is in contrast to $ z_{2}^{1} $, which is the sum bit of the \textit{second} generalized adder and is the lowest bit entry for the total sum of the assignment. All carry bits except for the very last one - the one from the single generalized adder in the last row - are intermediaries. Variables that are not between boxes are determined; $ s_{i}^{j} $ are set to the $ j $-th lowest bit in the $ i $-th integer of the set $ S $ while $ z_{end}^{i} $ for layer $ i $ and $ k_{1}^{5} $ are set to one in the corresponding matrix. 
	
	\subsection{The General Unreduced Case}\label{subsec:genunreduced}
	
	\quad Before we turn our attention to a full protocol for the general unreduced case, we give a more intuitive and visual description of the reduction. Fig.~\ref{fig:fullcircuit} gives a schematic of what the general case looks like. Note that Fig.~\ref{fig:examplecircuit}B fits precisely this description as well.
	
	\begin{figure}[h!]
		\centering
		\makebox[\textwidth][c]{\includegraphics[width=1.2\textwidth]{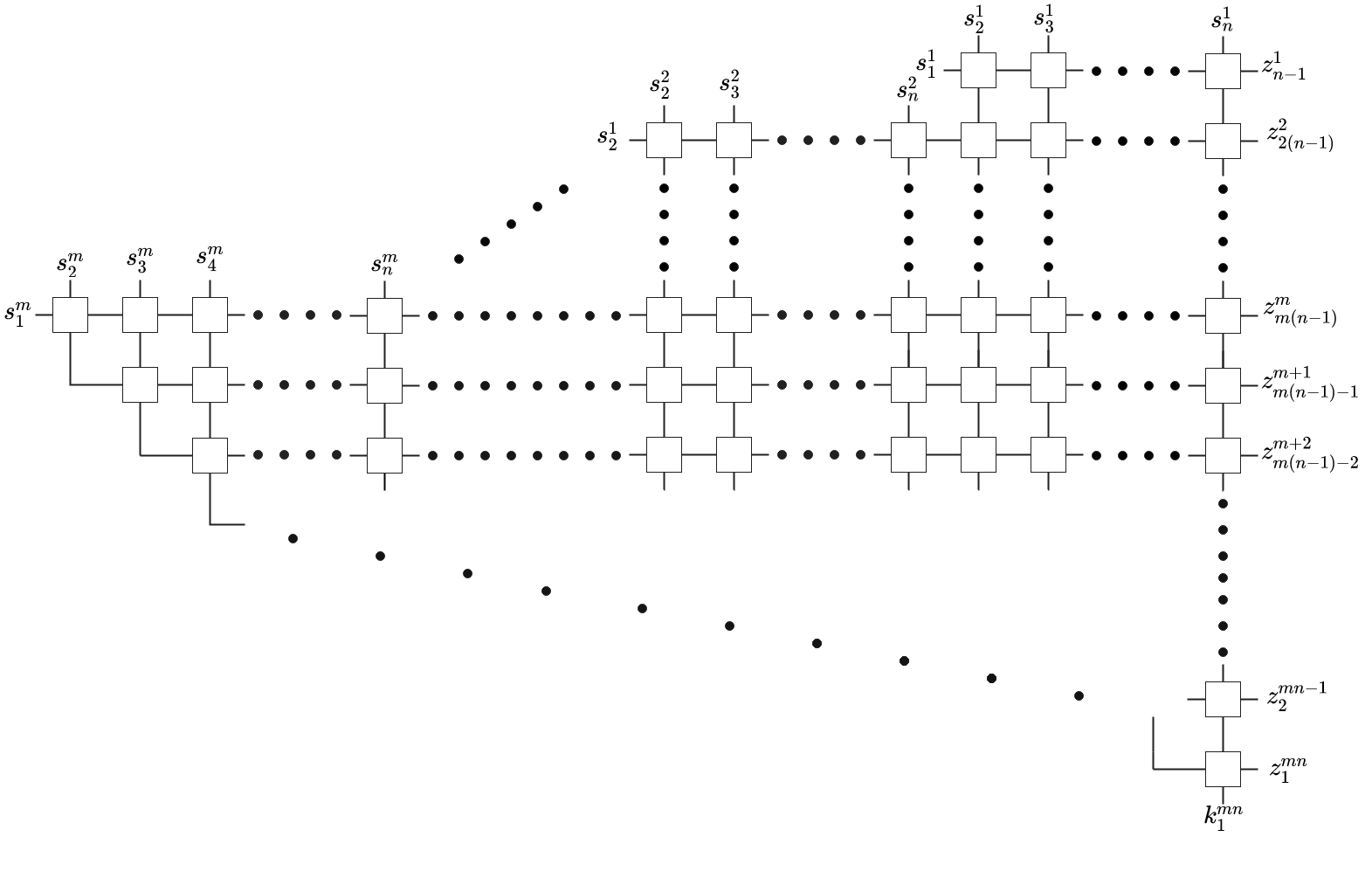}}
		\caption{This figure shows the layout of the generalized complete adder for enforcing that $ u^{\ast} $ is only valid if the corresponding $ u $ on $ S = \{ s_{1}, \ldots, s_{n} \} $ is also valid. In each row, a box corresponds to a generalized adder (with the truth table given in Table \ref{tab:genadd}) where the output of that whole row (labeled by the sum bit $ z $) is zero if and only if $ u^{\ast} $ is valid. After $ m $ (the largest bit length of any $ s_{i} \in S $) rows, the next rows are fed only carries from the previous rows. As such, the number of generalized adders decreases by one, until the very last row, where we have that $ z_{1}^{mn} $ and $ k_{1}^{mn} $ should both be zero for $ u^{\ast} $ to be valid on the set. The final constraint matrix is a representation of this diagram, with each generalized adder representing a submatrix that enforces the relationship shown in Table \ref{tab:genadd} (check Eq.~\ref{gamat}). }
		\label{fig:fullcircuit}
	\end{figure}
	
	\quad We call the generalized function with the truth table corresponding to Table \ref{tab:genadd} as $ GA_{s} $ and $ GA_{k} $ for the sum and carry bit respectively, and so the constraints we considered earlier enforce: $ u^{\ast}(c) = GA_{c}(u^{\ast}(a), u^{\ast}(b)) $ and $ u^{\ast}(s) = GA_{s}( u^{\ast}(a), u^{\ast}(b) ) $. We use the common convention of writing $ u^{\ast}( a_1, \ldots, a_k ) $ as a condensed form of $ ( u^{\ast}( a_1 ), \ldots, u^{\ast}( a_k ) ) $ so that $ GA_{c}( u^{\ast}( a ), u^{\ast}( b ) ) \equiv GA_{c}( u^{\ast}( a, b )) $. These are not proper functions since $ GA_{s} $ and $ GA_{k} $ sometimes have two valid modes of operation. We also define $ GA_{s}( u^{\ast}( s_{1:k} ) ) = GA_{s} \left( GA_{s} \left( \ldots GA_{s} \left( u^{\ast}( s_{1} ), u^{\ast}( s_{2} ) \right), \ldots \right), u^{\ast}( s_{n} ) \right) $ to help condense our writing.
	
	\m 
	
	\quad To enforce the right bit addition, we use the following protocol:
	\begin{problems}
		
		\item Let $ l = 1 $, $ \mathscr{K} = \O $, and $ \mathscr{A} = \O $.
		
		\item Generate $ (n-1)l $ carry bits ($ k_{1}^{l} , \ldots, k_{(n-1)l}^{l} $) and append them to $ \mathscr{A} $ as well as $ (n-1)l $  sum bits ($ z_{1}^{l}, \ldots, z_{(n-1)l}^{l} $) and append them to $ \mathscr{A} $. Add $ n-1 $ constraints to $ \mathscr{K} $ that will enforce $ GA $ between $ s_{1}^{l}, \ldots s_{n}^{l} $ in order, such that for any assignment $ u^{\ast} $, $ u^{\ast} $ is valid if and only if
		\begin{align}
		u^{\ast}( k_{1}^{l} )   &=  GA_{k}( u^{\ast}( s_{1}^{l} 	,	s_{2}^{l} ) )                                       \\
		u^{\ast}( z_{1}^{l} )   &=  GA_{s}( u^{\ast}( s_{1}^{l} 	,	s_{2}^{l} ) )                                       \\
		u^{\ast}( k_{i}^{l} )   &=  GA_{k}( u^{\ast}( z_{i-1}^{l}	,	s_{i+1}^{l} ) )	= GA_{k}( GA_{s}( u^{\ast}( s_{1:i}^{l} ) ) )   \;\; \forall i \in [2,n-1]      \\
		u^{\ast}( z_{i}^{l} )   &=  GA_{s}( u^{\ast}( z_{i-1}^{l}	,	s_{i+1}^{l} ) )	= GA_{s}( u^{\ast}( s_{1:i+1}^{l} ) )           \;\; \forall i \in [2,n-1]     
		\end{align}
		
		Then place $ n l $ constraints to $ \mathscr{K} $ that will enforce $ GA $ between $ \{ k_{1}^{l-1}, \ldots k_{(n-1)(l-1)}^{l-1} \} $ (the carry ins from the previous layer) and $ z_{n-1} $ such that for any assignment $ u^{\ast} $, $ u^{\ast} $ is valid if and only if:
		\begin{align}
		u^{\ast}( k_{i}^{l} )   &=  GA_{k}( u^{\ast}( z_{i-1}^{l}	,	k_{i-n}^{l-1} ) )     = GA_{k}( GA_{s}( u^{\ast}( s_{1:n}^{l}, k_{1:i-n}^{l-1} ) ) )  \;\; \forall i \in [n, (n-1)l]     \\
		u^{\ast}( z_{i}^{l} )   &=  GA_{s}( u^{\ast}( z_{i-1}^{l}	,	k_{i-n}^{l-1} ) )     = GA_{s}( u^{\ast}( s_{1:n}^{l}, k_{1:i-n}^{l-1} ) )             \;\; \forall i \in [n, (n-1)l]     
		\end{align}
		
		\item Let $ l = l + 1 $. If $ l \leq m $, then Go to Step 2.
		
		\item At the last run of step 2., we had $ (n-1)m $ total carry bits. Now we add layers feeding carries forward like before, but without introducing any new bits from the actual integers. As such, in each layer, we will have one less carry bit generated than the layer before it. 
		
		\item Let $ r = 1 $
		
		\item Generate $ (n-1)m - r $ carry bits $ \{ k_{1}^{r+m}, \ldots, k_{(n-1)m-r}^{r+m} \} $ and append them to $ \mathscr{A} $ as well as $ (n-1)m - r $ sum bits ($ z_{1}^{r+m}, \ldots, z_{(n-1)m-r}^{r+m} $) and append them to $ \mathscr{A} $. Add $ (n-1)m - r $ constraints to $ \mathscr{K} $ that will enforce GA on the carry bits of the previous layer. For the first layer:
		\begin{align}
		u^{\ast}( k_{1}^{m+1} ) &= GA_{k}( u^{\ast}( k_{1}^{m}, k_{2}^{m} 	))) 											\\
		u^{\ast}( z_{1}^{m+1} ) &= GA_{s}( u^{\ast}( k_{1}^{m}, k_{2}^{m} 	))) 											\\
		u^{\ast}( k_{i}^{m+1} ) &= GA_{k}( u^{\ast}(	z_{i}^{m+1}, k_{i+1}^{m}	)) 		\;\; \forall i \in [2,m(n-1)-1]     \\
		u^{\ast}( z_{i}^{m+1} ) &= GA_{k}( u^{\ast}(	z_{i}^{m+1}, k_{i+1}^{m}	)) 		\;\; \forall i \in [2,m(n-1)-1]     
		\end{align}
		
		For all the subsequent layers:
		\begin{align}
		u^{\ast}( k_{1}^{m+r} ) &= GA_{k}( u^{\ast}( k_{1}^{m+r-1}	, 	k_{2}^{m+r-1} 	))) 										\\
		u^{\ast}( z_{1}^{m+r} ) &= GA_{s}( u^{\ast}( k_{1}^{m+r-1}	, 	k_{2}^{m+r-1} 	))) 										\\
		u^{\ast}( k_{i}^{m+r} ) &= GA_{k}( u^{\ast}( z_{i}^{m+r} 	,	k_{i+1}^{m+r-1} )) = GA_{k}( GA_{s} ( u^{\ast}( k_{1:i+1}^{m+r-1} ) ) )   \;\; \forall i \in [2,m(n-1)-r]     \\
		u^{\ast}( z_{i}^{m+r} ) &= GA_{s}( u^{\ast}( z_{i}^{m+r} 	,	k_{i+1}^{m+r-1} )) = GA_{s}( u^{\ast}( k_{1:i+1}^{m+r-1} ))	\;\; \forall i \in [2,m(n-1)-r]	
		\end{align}
		
		\item Let $ r = r + 1 $. If $ r \leq m(n-1) $, go to Step 6. 
		
		\item Lastly add constraints to force the last sum bit in each row to be zero, those constraints simply are $ \{ z_{(n-1)l}^{l}\} $ for $ l \in \{ 1, \ldots, m \} $ and $ \{ z_{m(n-1)-r}^{m+r} \} $ for $ r \in \{ 1, \ldots, m(n-1) \} $ (check Fig.~\ref{fig:fullcircuit}). We also add $ \{ k_{1}^{m(n-1)} \} $.

	\end{problems}
	
	\begin{thm}
		Suppose there exists $ u $ such that $ \sum_{i=1}^{n} u_{i} s_{i} = 0 $, then and only then does there exist $ u^{\ast} $ such that $ u^{\ast}( z_{end}^{l} ) = u^{\ast}( k_{1}^{mn} ) = 0 $ (where $ z_{end}^{l} $ refers to the last sum bit in each row as shown in \ref{fig:fullcircuit}). Then $ E_{S}( u( S ) ) = 0 \mim E_{S \cup \mathscr{A} } ( u^{\ast}( S \cup \mathscr{A} ) ) = 0 \mim K^{M} \, \vec{\mu^{\ast}} = 0 $.
	\end{thm}
	
	\begin{proof}
		
		\quad We first consider the forward direction. First recognize that $ \sum_{i=1}^{n} u_{i} s_{i} = \sum_{i=1}^{n} \sum_{j=1}^{m} u_{i} s_{i}^{j} 2^{j} $. It must be that $ \sum_{i=1}^{n} u_{i} s_{i}^{1} \mod 2 = 0 $. Then $ \sum_{i=1}^{n} u_{i} s_{i}^{1} = \sigma^{1} \in \{ \ldots, -4, -2, 0, 2, 4, \ldots \} $. Note that if inputs $ a $ and $ b $ have different signs for the generalized adder, the carry and sum bits are both zero, if $ a $  and $ b $ are the same sign then they pass a carry. When one is zero, then the other one is simply passed on using the primary operation of $ GA_{k} $ and $ GA_{s} $. In the forward direction of the proof, we only need to consider the primary operations. As such, it is clear that $ u^{\ast}( z_{n-1}^{1} ) = 0 $ since the number of positive and negative inputs added is zero modulo 2. It should also be straight forward to see that $ \sum_{i=1}^{n-1} u^{\ast}( k_{i}^{1} ) = \frac{ \sigma^{1} }{ 2 } $. Now recognize that $ \left( \sum_{i=1}^{n} \sum_{j=1}^{l} u_{i} s_{i}^{j} \right) / 2^{l} = \sigma^{l} \in \{ \ldots, -4, -2, 0, 2, 4, \ldots \} $. Note that $ \sigma^{l} = \frac{ \sigma^{l-1} }{ 2 } + \sum_{i=1}^{n} u_{i} s_{i}^{l} $ where we can identify $ \frac{ \sigma^{l-1} }{ 2 } = \sum_{i=1}^{f[l-1]} u^{\ast}( k_{i}^{l-1} ) $, with $ f[x] = \Theta(m-x)x(n-1) + \Theta(x-m)(m(n-1)-x) $ for all $ i $. Here $ f[x] $ has a Heaviside step function to differentiate between the indexing of rows generated by Step 2 of the protocol versus those generated later by Step 6. Again it is clear that $ u^{\ast}( z_{f(l)}^{l} ) = 0 $ since the number of positive inputs and negative inputs of $ u^{\ast}( s_{1:n}^{l}, k_{1:f(l-1)}^{l-1} ) $ is zero modulo 2. Since $ \sum_{i=1}^{n} | s_{i} | < n 2^{m} $, we must only worry at most about $ m \log( n ) $ rows, but we have $ mn $ rows as zero for $ u^{\ast} $.	
		
		\m 
		
		\quad We now consider the backward direction. The proof will look very similar to the forward direction, but now we also have to give some consideration that $ u^{\ast} $ could make use of \textit{secondary} operations, not just \textit{primary} operations. Consider in a specific row, we used the secondary operations, e.g. $ \tilde{GA}_{k}( 1, 0 ) = 1 $ and $ \tilde{GA}_{s}( 1, 0 ) = -1 $. Here we used the tilde to alert the reader that these are the secondary operations specifically. We know that $ u^{\ast}( z_{end}^{l} ) = 0 $ for any layer $ l $ as the assumption, and so the number of $ \tilde{GA} $ operations is even. It must be of opposite kinds such that the number of total carries is unchanged (since they are still valid operations such that $ 2 \, c + s = a + b $) - for every operation that propagrates an extra carry at the expense of reducing its sum bit there must be a secondary operation that reduces its carry bit to surplus its sum bit. If not, then $ u^{\ast}( z_{end}^{l} ) \neq 0 $. Then we can replace them with the primary operations. The rest of the arguments follow through as before. We have $ u^{\ast}( z_{n-1}^{l} ) = 0 $ and since $ \sum_{i=1}^{n} u^{\ast}( s_{i}^{1} ) = \sum_{i=1}^{n-1} 2 \times u^{\ast}( k_{i}^{1} ) + u^{\ast}( z_{n-1}^{1} ) $ with $ u^{\ast}( z_{n-1}^{1} ) = 0 $, we have $ \sigma^{1} = \sum_{i=1}^{n-1} u^{\ast}( k_{i}^{1} ) $. Again $ \sigma^{l} + u^{\ast}( z_{end}^{l} ) = \frac{ \sigma^{l-1} }{ 2 } + \sum_{i=1}^{n} u^{\ast}( s_{i}^{l} ) $ and we know that $ u^{\ast}( z_{end}^{l} ) = 0 $. By the same bound, we know that after $ mn $ rows having zero on all the outputs is sufficient to see that $ \sum_{i=1}^{n} \sum_{j=1}^{m} u^{\ast}( s_{i}^{j} ) = 0 $ and so let $ u( s_{i} ) = u^{\ast}( s_{i} ) $ for every $ s_{i} \in S $. 
	\end{proof}
	
	\m 
	
	Given an input integer set $ S $, this protocol outputs constraint set $ \mathscr{K} = \{ K_{1}, \ldots, K_{S\cup \mathscr{A}} 
	\} $ (the variables with indices such that the entry is nonzero in $ K^{M} $) such that an assignment vector $ \vec{ \mu^{\ast} } $ has value zero for every constraint in the set if and only if these exists a valid assignment for $ S $ that defines two disjoint nonempty sets $ A $ and $ B $ such that $ \sum_{a \in A} a - \sum_{b \in B} b = 0 $. We can then identify the constraint operators as the row vectors of $ K^{M} $ as coefficients on spin-z operators on each qubit: $ \hat{K}_{i} = \sum_{j}^{|S \cup \mathscr{A} |} k_{ij}^{m} \sigma_{j}^{z} $, such that our solver for \texttt{0-1-LP-QCOMMUTE} finds a Hermitian matrix that commutes with these constraint operators.
	
	\m 
	
	\subsection{Proof of Runtime}
	
	\quad In the worse case, we construct no more than $ 5 $ new variables and $ 6 $ constraints for each $ GA $ illustrated in Fig.~\ref{fig:fullcircuit}. For row $ i < m+1 $, this leads to no more than $ 5 \, (n-1) \, i $ new variables and $ 6 \, n \, (n-1) \, i $ constraints. Then after row $ m $, we have no more than $ 5 \, (n-1) \, m^{2} $ variables and $ 6 \, (n-1) \, m^{2} $ constraints. Row $ i > m $ has no more than $ m \, n - i $ generalized adders, creating no more than $ 5 \, ( m \, n - i ) $ new variables and $ 6 \, ( m \, n - i ) $ constraints. In total, we have no more than $ \mathscr{O} \left( m^{2} n^{2} \right) $ variables and $ \mathscr{O} \left(m^{2} n^{2} \right) $ constraints. The constraint matrix therefore has size $ \mathscr{O} \left( m^{2} n^{2} \right) \times \mathscr{O} \left( m^{2} n^{2} \right) $. The reduction is therefore a polynomial time algorithm.
	
	\subsection{Reducing a Solution of \texttt{ILP-COMMUTE} to a Solution of \texttt{EQUAL SUBSET SUM}}\label{sec:appendsoltosol}

    \quad We consider the same set up as in Section~\ref{sec:probdef}. Using the protocol from Section~\ref{sec:redoxold} (check Fig.~\ref{fig:reduxflowchart}) we can reduce any instance of the \texttt{EQUAL SUBSET SUM} problem with polynomial overhead, and if any solution to \texttt{0-1-LP-QCOMMUTE} exists, there must be $ v $ and $ w $ (and therefore the sets $ A $ and $ B $) to describe at least one off-diagonal term in the spin-z basis. By our construction, the ancilla bits used for \textit{forcing} are the bits beyond the n-th bit. Similarly, if a solution to \texttt{EQUAL SUBSET SUM} exists, by selecting the values of $ v $ and $ w $ to match the indices of chosen elements for the sets $ A $ and $ B $, we are able to set the values of the first $ n $ bits and then propagate their value through the general adder to find $ v' $ and $ w' $ over the enlarged space. Then $ H = \bigotimes_{i=1}^{|S\bigcup\mathcal{A}|} \left( \sigma^{+} \right)^{ v_{i}' } \left( \sigma^{-} \right)^{ w_{i}' } + \bigotimes_{i=1}^{|S\bigcup\mathcal{A}|} \left( \sigma^{-} \right)^{ v_{i}' } \left( \sigma^{+} \right)^{ w_{i}' } $ solves \texttt{0-1-LP-QCOMMUTE}. This leads to the following Theorem:
	
	\begin{thm}
		\texttt{0-1-LP-QCOMMUTE} is NP-Hard. 
	\end{thm}
	
    Through the same proof that \texttt{ILP-QCOMMUTE} is polynomial verifiable, \texttt{0-1-LP-QCOMMUTE} is likewise polynomial verifiable.
    
    \begin{thm}
        \texttt{0-1-LP-QCOMMUTE} is NP-Complete.
    \end{thm}

	\quad It also leads to an important corollary:
	
	\begin{cor}
		$ \{ -1, 0, 1 \} $\texttt{-LP-QCOMMUTE} is NP-Complete.
	\end{cor}
	
	\quad As well as another proof to the result in Section~\ref{sec:probdef}:
	
	\begin{cor}
	    \texttt{ILP-QCOMMUTE} is NP-Complete.
	\end{cor}
	
	\section{Proof of the Matrix Implementation of the Generalized Adder}\label{proofmatrix}
	
	\quad Given inputs $ a $ and $ b $, we define the matrix on $ a, b, s, c, x_{1}, x_{2}, x_{3} $ - with $ x_{1}, x_{2}, x_{3} $ being intermediating ancillas - as:
	\begin{align}
        GA^{M} = \; \begin{array}{c c c c c c c c c}
                    &   a & b & x_1 & x_2 & x_3 & c & s &     \\ 
\ldelim({6}{0.5em}  &   a  &  b  &   1  &   1   &   1   &   0   &   0   & \hspace{-0.2em}\rdelim){6}{0.5em}  \\
	                &   0  &  0  &   1  &   0   &   1   &   1   &   1   &	\\
	                &   0  &  0  &   0  &   1   &   1   &   1   &   1   &	\\
	                &   0  &  0  &   1  &   0   &   0   &   1   &   0   &	\\
	                &   0  &  0  &   0  &   1   &   0   &   1   &   0	&	\\
	                &   0  &  0  &   0  &   0   &   1   &   0   &   1	&	
                    \end{array}
	\end{align}

	As constraints, we can write it as:
	\begin{align}
	GA_{1}( a, b, x_{1}, x_{2}, x_{3} ) &= 0, 	    \\
	GA_{2}( x_{1}, x_{3}, s, c ) &= 0, 			    \\
	GA_{3}( x_{2}, x_{3}, s, c ) &= 0, 			    \\
	GA_{4}( x_{1}, c ) &= 0, 						\\
	GA_{5}( x_{2}, c ) &= 0, 						\\
	GA_{6}( x_{3}, s ) &= 0.
	\end{align}
	
	\quad For every generalized adder in Fig.~\ref{fig:fullcircuit} (as described in the protocol we gave in Section~\ref{subsec:genadd}), we have a submatrix over the corresponding variables. We give a simple case by case proof that $ GA^{M} $ enforces $ u^{\ast} $ to be valid if and only if its entries satisfy $ 2 \, u^{\ast}( c ) + u^{\ast}( s ) = u^{\ast}( a ) + u^{\ast}( b ) $ as seen in Fig.~\ref{tab:genadd}.
	
	\quad A constraint is satisfied if and only if the assignment $ u^{\ast} $ over the variables of that constraint sums to zero. Then an assignment satisfies all of them if $ u^{\ast}(GA_{i}) = 0 $ for all $ i \in [1,6] $. Although checkable through brute force calculations, we give simple arguments for emulation of bit addition step by step:
	
	\m
	
	\quad If $ u^{\ast}(a) = u^{\ast}(b) = 1 $, then and only then do we have $ u^{\ast}(c) = 1 $ and $ u^{\ast}(s) = 0 $. If $ u^{\ast}(a) = 1 $ and $ u^{\ast}(b) = 1 $, then two of the auxillary bits must have an assignment of -1 and one must not. If $ u^{\ast}(x_{1}) = -1 $ then $ u^{\ast}(c) = 1 $ by $ GA_{4} $, but then $ u^{\ast}(x_{2}) = -1 $ by $ GA_{5} $ and vice versa. Then $ u^{\ast}(x_{3}) = 0 $, otherwise $ GA_{1} $ cannot be satisfied. Then $ u^{\ast}(s) = 0 $ as wanted. Suppose that $ u^{\ast}(s) = 0 $ and $ u^{\ast}(c) = 1 $, then likewise $ GA_{4} $ and $ GA_{3} $ force that $ u^{\ast}(x_{1}) = u^{\ast}(x_{2}) = -1 $. Then from $ GA_{2} $, we have that $ u^{\ast}(x_{3}) = 0 $ and so from $ GA_{1} $ that $ u^{\ast}(a) = u^{\ast}(b) = 1 $.
	
	\m 
	
	\quad If $ u^{\ast}(a) = 1 $ and $ u^{\ast}(b) = 0 $ or $ u^{\ast}(a) = 0 $ and $ u^{\ast}(b) = 1 $, then and only then do we have $ u^{\ast}(c) = 0 $ and $ u^{\ast}(s) = 1 $ or $ u^{\ast}(c) = 1 $ and $ u^{\ast}(s) = -1 $. Suppose that $ u^{\ast}(a) = 1 $ or $ u^{\ast}(b) = 1 $, but not both. From $ GA_{1} $, we know that either one of the auxillary bits must take value -1 or two take the value -1 and one takes the value 1. If either $ x_{1} $ or $ x_{2} $ take value -1, but not the other then $ GA_{4} $ and $ GA_{5} $ lead to a contradiction, then if $ x_{3} = 1 $ and by $ GA_{6} $ it must be that $ s = -1 $. Otherwise $ x_{1} = x_{2} = 0 $ and so $ u^{\ast}(x_{3}) = -1 $ by $ GA_{2} $ and $ GA_{3} $. $ GA_{6} $ forces that $ u^{\ast}(s) = 1 $. Suppose that $ u^{\ast}(c) = 0 $ and $ u^{\ast}(s) = 1 $. Then $ u^{\ast}(x_{3}) = -1 $ from $ GA_{6} $ and $ u^{\ast}(x_{1}) = u^{\ast}(x_{2}) = 0 $ from $ GA_{4} $ and $ GA_{5} $. Then $ GA_{1} $ is only satisfied if $ u^{\ast}(a) = 1 $ or $ u^{\ast}(b) = 1 $, but not both. Suppose instead that $ u^{\ast}(c) = 1 $ and $ u^{\ast}(s) = -1 $. Then by $ GA_{6} $ it must be that $ x_{3} = 1 $. By $ GA_{5} $ and $ GA_{4} $, it must be that $ x_{1} = -1 $ and $ x_{2} = -1 $. Then by $ GA_{1} $ it must be that $ u^{\ast}(a) $ or $ u^{\ast}(b) $ is 1, but not both.
	
	\m 
	
	\quad If $ u^{\ast}(a) + u^{\ast}(b) = 0 $, then and only then do we have $ u^{\ast}(c) = 0 $ and $ u^{\ast}(s) = 0 $. Suppose that $ u^{\ast}(a) + u^{\ast}(b) = 0 $. From $ GA_{1} $, we know that at most two auxillary bits are non-zero and they have opposite sign. From $ GA_{4} $ and $ GA_{5} $ if one of the first two auxillary bits is non-zero then so is the other one, but they must have the same sign. As such $ GA_{1} $ can only be satisfied with $ u^{\ast}(x_{1}) = u^{\ast}(x_{2}) = u^{\ast}(x_{3}) = 0 $. Then it follows that $ u^{\ast}(c) = u^{\ast}(s) = 0 $. Suppose that $ u^{\ast}(c) = u^{\ast}(s) = 0 $. From $ GA_{4} $, $ GA_{5} $, and $ GA_{6} $, we have that $ u^{\ast}(x_{1}) = u^{\ast}(x_{2}) = u^{\ast}(x_{3}) = 0 $. Then $ GA_{1} $ can only be satisfied if $ u^{\ast}(a) + u^{\ast}(b) = 0 $.
	
	\m
	
	\quad The same logic works if we swap the values $ 1 $ and $ -1 $ everywhere in the above proof.

\end{document}